\documentclass[10pt,journal,compsoc]{IEEEtran}
\usepackage{graphicx}
\usepackage{multicol}
\usepackage{enumitem}
\usepackage{msc}
\usepackage{color}    
 \usepackage{booktabs}

\usepackage{amsthm}

\usepackage[colorlinks, citecolor=blue]{hyperref}

\usepackage{algpseudocode}
\usepackage{xcolor}
\include{pythonlisting}
\usepackage{fancyhdr,graphicx,amsmath,amssymb}
\usepackage[ruled,vlined]{algorithm2e}

\newtheorem{definition}{Definition}[section]
\newtheorem{theorem}{Theorem}[section]
\newtheorem{lemma}{Lemma}[section]
\newtheorem{remark}{Remark}[section]
\usepackage{pifont}
\newcommand{\cmark}{\ding{51}}%
\newcommand{\xmark}{\ding{55}}%

\newcommand{\strlen}{\ensuremath{\textsf{len}}\xspace}
\newcommand{\roundr}{\ensuremath\mathcal R_{d}\xspace}
\newcommand{\rounds}{\ensuremath\aleph\xspace}

\newcommand{\chainlength}{\ensuremath{\mathcal{L}}\xspace}

\newcommand{\Cq}{\ensuremath{C_{q}}\xspace}

\newcommand{\astkratio}{\ensuremath{\Upsilon}\xspace}

\newcommand{\sysp}{\ensuremath{\epsilon}\xspace}

\newcommand{\chainC}{\ensuremath{\mathsf{C}}\xspace}

\usepackage{graphicx}
\usepackage{textcomp}
\usepackage{xcolor}

\ifCLASSOPTIONcompsoc
  \usepackage[nocompress]{cite}
\else
  \usepackage{cite}
\fi
\ifCLASSINFOpdf
\else
\fi
\begin{document}

\SetKwComment{Comment}{$\triangleright$\ }{}

%
\title{Reinshard: An optimally sharded dual-blockchain for concurrency resolution}
%
%
%
%

\author{Vishal Sharma,
        Zengpeng Li,
        Pawel Szalachowski,
        Teik Guan Tan,
        Jianying Zhou 
\IEEEcompsocitemizethanks{\IEEEcompsocthanksitem V. Sharma is with the School of Electronics, Electrical Engineering and Computer Science (EEECS), Queen's University Belfast (QUB), Northern Ireland, United Kingdom. \protect
E-mail: v.sharma@qub.ac.uk
\IEEEcompsocthanksitem Z. Li is with the School of Cyebr Science and Technology, Shandong University, Qingdao, China. \protect
E-mail: zengpengliz@gmail.com
\IEEEcompsocthanksitem P. Szałachowski, T.G. Tan, and J. Zhou are with the ISTD Pillar, Singapore University of Technology and Design, Singapore. \protect
E-mail: pjszal@gmail.com, teikguan\_tan@mymail.sutd.edu.sg, jianying\_zhou@sutd.edu.sg}
\thanks{Manuscript received.}
}

\IEEEtitleabstractindextext{%
\begin{abstract}
Decentralized control, low-complexity, flexible and efficient communications are the requirements of an architecture that aims to scale blockchains beyond the current state. Such properties are attainable by reducing ledger size and providing parallel operations in the blockchain. Sharding is one of the approaches that lower the burden of the nodes and enhance performance. However, the current solutions lack the features for \textit{resolving concurrency} during cross-shard communications. With multiple participants belonging to different shards, handling concurrent operations is essential for optimal sharding. This issue becomes prominent due to the lack of architectural support and requires additional consensus for cross-shard communications. Inspired by hybrid \textit{Proof-of-Work/Proof-of-Stake} (PoW/PoS), like \textit{Ethereum}, \textit{hybrid consensus} and \textit{2-hop blockchain}, we propose \textit{Reinshard}, a new blockchain that inherits the properties of hybrid consensus for optimal sharding. \textit{Reinshard} uses PoW and PoS chain-pairs with PoS sub-chains for all the valid chain-pairs where the hybrid consensus is attained through \textit{Verifiable Delay Function} (VDF). Our architecture provides a secure method of arranging nodes in shards and resolves concurrency conflicts using the delay factor of VDF. The applicability of \textit{Reinshard} is demonstrated through security and experimental evaluations. A practical concurrency problem is considered to show the efficacy of \textit{Reinshard} in providing optimal sharding.
\end{abstract}

\begin{IEEEkeywords}
Blockchain, Hybrid Consensus, Sharding, Concurrency Control.
\end{IEEEkeywords}}

\maketitle

\IEEEdisplaynontitleabstractindextext

%
\IEEEpeerreviewmaketitle

\IEEEraisesectionheading{\section{Introduction}\label{sec:introduction}}
Blockchain through its decentralized and authenticated append-only ledgers provides transparent, easily available, and accessible digital content sharing with censorship resistance~\cite{DBLP:conf/ndss/MalavoltaMSKM19,xu2019blockchain,lei2020groupchain}. The use of blockchain is dominated by the cryptocurrencies in the financial market and major efforts are made to make them efficient, scalable, and flexible~\cite{DBLP:conf/nsdi/WangW19,ghorbanian2020methods}. However, the market trends are changing and the applications of blockchains are seen beyond cryptocurrencies~\cite{DBLP:conf/icbc2/MarsalekKZT19}. The maintenance of low-latency, flexible, secure and scalable architecture is the key requirement of blockchain applications\cite{li2018blockchain,fan2020secure}. There exist a plethora of approaches to the formation of provably secure and scalable blockchains by using different consensus like PoW, PoS, etc. Dominantly, PoS has been a promising paradigm for future blockchain implementations as it provides multiple benefits over PoW-based systems, which involve a highly difficult puzzle-solving increasing with the network size. In general, PoS can be divided into four main categories~\cite{Xiao2019}, namely, chain-based PoS (Ppcoin~\cite{king2012ppcoin}, Nxt~\cite{Nxt2014}), committee-based PoS (Ouroboros~\cite{kiayias2017ouroboros}, Ouroboros Praos~\cite{david2018ouroboros}, Snow White~\cite{daian2017snow}), Byzantine Fault Tolerance (BFT)-based PoS (Tendermint~\cite{kwon2014tendermint}, Algorand~\cite{gilad2017algorand}), and delegation-based PoS (Lisk~\cite{lisk}, EOS.IO~\cite{eosio2018}). However, despite clear benefits, PoS-blockchains need to counter with issues affecting the consensus, which include influential assets, identity revealing, unwanted centralization, non-flexible and non-adaptive chaining, and community-frauds~\cite{Xiao2019}.
Scaling blockchains beyond the current state require decentralized control, low-complexity, flexible and efficient communications, which can be attained by reducing ledger size and providing parallel operations in the blockchain, such as \textit{Sharding}\cite{xie2019survey}.

In the sharded-blockchains, the communications are severely affected by \textit{concurrent requests} between the shards, which includes the split infrastructure of the blockchain. Concurrency-resolution requires enough delay to initiate a waiting mechanism between the involved parties to complete their transactions. A solution in the form of new architecture is required that can provide this delay during the cross-shard communications without compromising the security as a higher delay value can risk the security. Such architecture should be supported by an efficient consensus scheme that should not affect the general workflow as well as maintain the security properties of the sharded-blockchain that involves chain-growth, chain-quality, common-prefix, and unbiased sharding. The main \textit{contributions} of our work are:
\begin{itemize}
\setlength\itemsep{0.01em}
    \item A dual-blockchain architecture, inspired by hybrid consensus~\cite{Duong2017,DBLP:conf/wdag/PassS17,Hcashteam2017}, is presented with provision of auto-sharding.
    \item The new architecture resolves the fundamental issue of concurrency conflicts during cross-shard communications. A simulated VDF is used for resolving \emph{concurrency issues} and chain-extensions via a delay-factor.
    \item An exemplary concurrency problem, \textit{Train}-and-\textit{Hotel} booking system~\cite{Tribble}, is used for demonstrating the effectiveness of the proposed model.
  \end{itemize}
\section{Related Works}\label{background}
Security and decentralization are two main factors to maintain in large scale and permissionless systems~\cite{Xiao2019,Syta2017}. The existing mechanisms use a central entity or checkpoints for attack mitigation. However, this leads to unwanted centralization and affects performance, which can be handled through optimal sharding. It involves both intra- and inter-(cross) shard communications. Several sharding mechanisms have been developed for either scalability, security, or both. Most prominent includes, Zilliqa's model~\cite{zilliqa2017zilliqa,VanValkenburgh2016}, Omniledger~\cite{kokoris2018omniledger}, RapidChain~\cite{Zamani2018}, Elastico~\cite{luu2016secure}, and Harmony~\cite{VanValkenburgh2016}. However, there exists a practical issue of concurrency when cross-shard operations are involved, and there are no evident demonstrations on how \textit{concurrency} is handled in these sharded blockchains. Such issues cause deadlocks which prohibits the use of blockchain for a wide range of network-based applications. The lack of architectural-support by the existing blockchain-solutions in resolving concurrency conflicts while performing inter-shard communications motivated us to design a new blockchain without compromising the security.
\begin{figure}[!ht]
  \centering
  \includegraphics[width=220px]{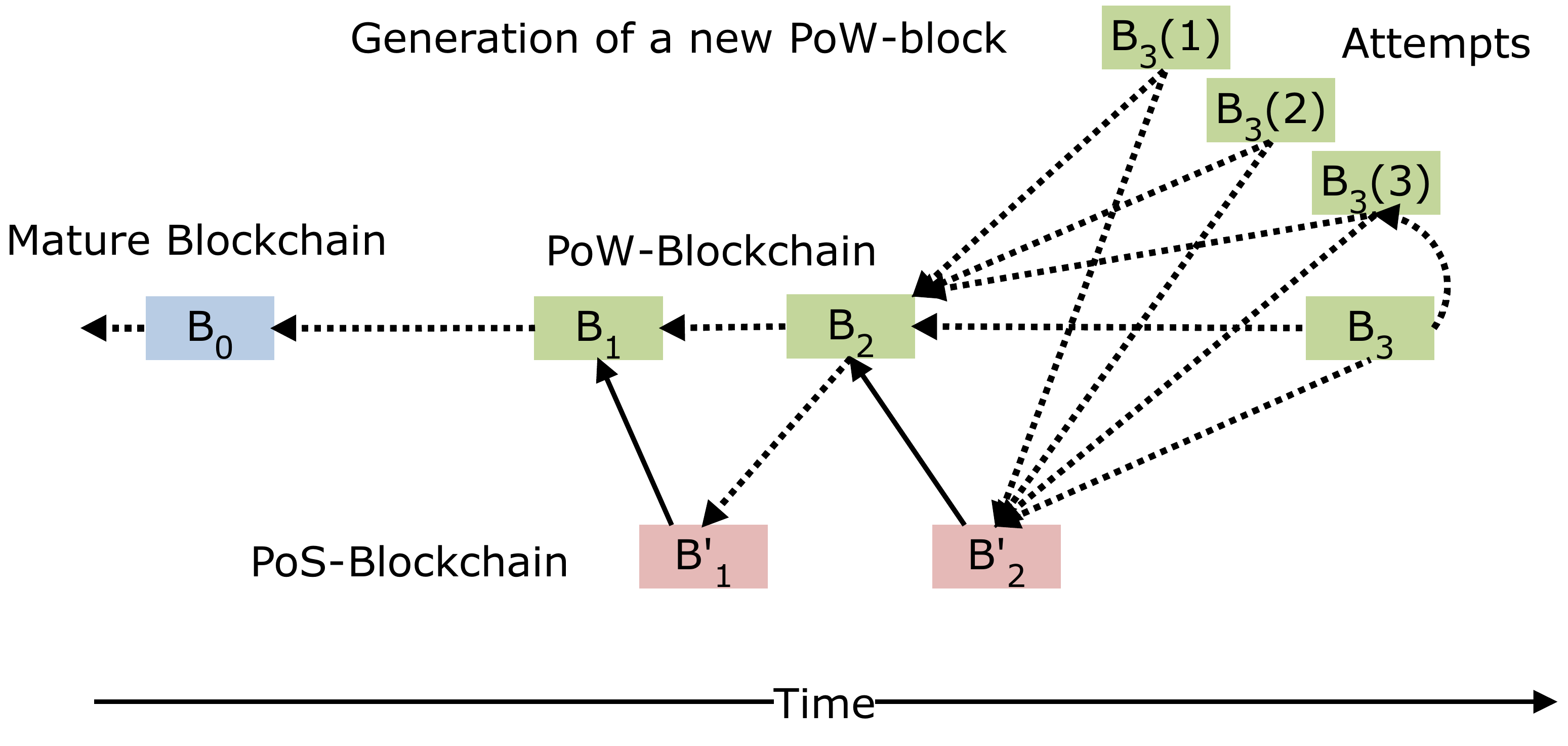}
  \caption{An illustration of modified 2-hop blockchain structure for generating PoW-blocks in TwinsCoin~\cite{Duong2018}}\label{figure2}
\end{figure}

Hybrid consensus allows better support for integrating, securing and improving different blockchains by using the strength of one consensus mechanism to resolve the shortcomings of another~\cite{DBLP:conf/wdag/PassS17,Hcashteam2017,DBLP:journals/corr/abs-1710-09437}. One of the examples of a hybrid blockchain is 2-hop blockchain~\cite{Duong2017}, which is a provably secure and scalable public blockchain. It uses PoW and PoS chain-pair as its operational principles. The key principle is two-way security which prevents an attacker from controlling the PoW-chain even if it owns the majority of the mining power. The use of PoS along with PoW allows the construction of a secure chain which is extended as TwinsCoin by Duong et al.~\cite{Duong2018}. The primary advantage of 2-hop blockchain and TwinsCoin is the prevention of adversaries that tend to control more than 50\% of the mining power~\cite{Duong2018,Duong2017}. Fig.~\ref{figure2} shows the PoW-block generation in TwinsCoin by following the concept of a modified 2-hop blockchain.

The difference between the 2-hop blockchain~\cite{Duong2017} and TwinsCoin~\cite{Duong2018} is the variation in the difficulty adjustments. The former uses an equal number of blocks for PoW and PoS chain, which makes difficulty-adjustment tedious and high complex; whereas the latter uses a ratio mechanism to decide on the number of PoW and PoS blocks, as shown in Table~\ref{Table1}. This makes difficulty-adjustment possible, however at an expense of additional resources which \textit{are not yielding to performance apart from security}. The 2-hop blockchain uses a PoW-PoS chain-pair, in which an extension to PoW chain is carried using a hash of the previous PoW block, a hash of the head of the PoS block and then solving it to find a suitable solution until a value lesser than the target value is not attained~\cite{Duong2017}. The PoS-chain is extended by finding a verification key to satisfy the new block under the given PoS target. This mechanism is used to generate a long-chain in one direction and then select the best valid chain-pair.
\begin{table*}[!ht]
\fontsize{8}{10}\selectfont
\centering
\caption{Comparison of chain extensions and difficulty adjustments of existing blockchain models. $h_{\rho} \in \{0,1\}^{\lambda}$ is the hash of the previous PoW block, $\lambda$ is the security parameter, $\rho$ is the suitable solution, $X$ is the block record, $T$ is the current PoW target, $h_{s} \in \{0,1\}^{\lambda}$ is the hash of the head of the PoS chain, $B_{PoW}$ is the PoW block, $v_{k}$ is the verification key, $\widetilde{T}$ is the current PoS target, $h_{a} \in \{0,1\}^{\lambda}$ is hash of the attempting PoW block for the local view, $r$ is the epoch, $t$ is the expected time of an epoch, $t_{r}$ is the actual time of an epoch, $n$ is the number of coins with the stakeholder, $E$ is the expectation of probability of successful PoW to stakeholder mapping, $\mu_{r}$ are the PoW blocks in $r$th epoch and $\mu$ is the number of blocks after which the difficulty has to be adjusted. (NA: Not Applicable because of principle difference)}\label{Table1}
\begin{tabular}{lllll}
\hline
\textbf{Blockchain} & \textbf{Extensions (PoW)} & \textbf{Extensions (PoS)} & \textbf{Difficulty (PoW)} & \textbf{Difficulty (PoS)} \\\hline
\hline\\
Nakamoto blockchain~\cite{nakamoto2008bitcoin}   &    $H(h_{\rho}, \rho, X) < T$              & $NA$             &           $T_{r+1}=\frac{t_{r}}{t}T_{r}$       &            $NA$      \\\\\hline\\
2-hop blockchain\cite{Duong2017}    &     $H(h_{\rho}, h_{s}, \rho) < T$             &    $\widetilde{H}(B_{PoW}, v_{k})<\widetilde{T}$               &       $NA$          & $NA$\\\\\hline\\
TwinsCoin~\cite{Duong2018}  &       $H(h_{\rho}, h_{s}, h_{a}, \rho)<T$          &                $\underbrace{\widetilde{H}(B_{PoW}, v_{k})<n\widetilde{T}}_\text{non-flat model}$  &   $T_{r+1}=\frac{\mu\;t_{r}}{\mu_{r}\;t\;E}T_{r}$               &        $\widetilde{T}_{r+1}=\frac{\mu_{r}\;E}{\mu}\widetilde{T}_{r}$      \\\\\hline
\end{tabular}
\end{table*}
In 2-hop blockchain, a PoW-block must be accompanied by a PoS-block for its possible inclusion in the blockchain, whereas in TwinsCoin, PoW-blocks without PoS-blocks are referred to as attempting blocks. The actual length of the blockchain is only based on the successful blocks.

\section{Proposed Model: Reinshard}\label{proposedmodel}
The proposed blockchain is a combination of the PoW-PoS chain-pair and PoS sub-chain which allows an effective way of including blocks, especially targeted towards the resolution of the concurrency issues, as shown in Fig.~\ref{figure3}.
\begin{figure}[!ht]
  \centering
  \includegraphics[width=240px]{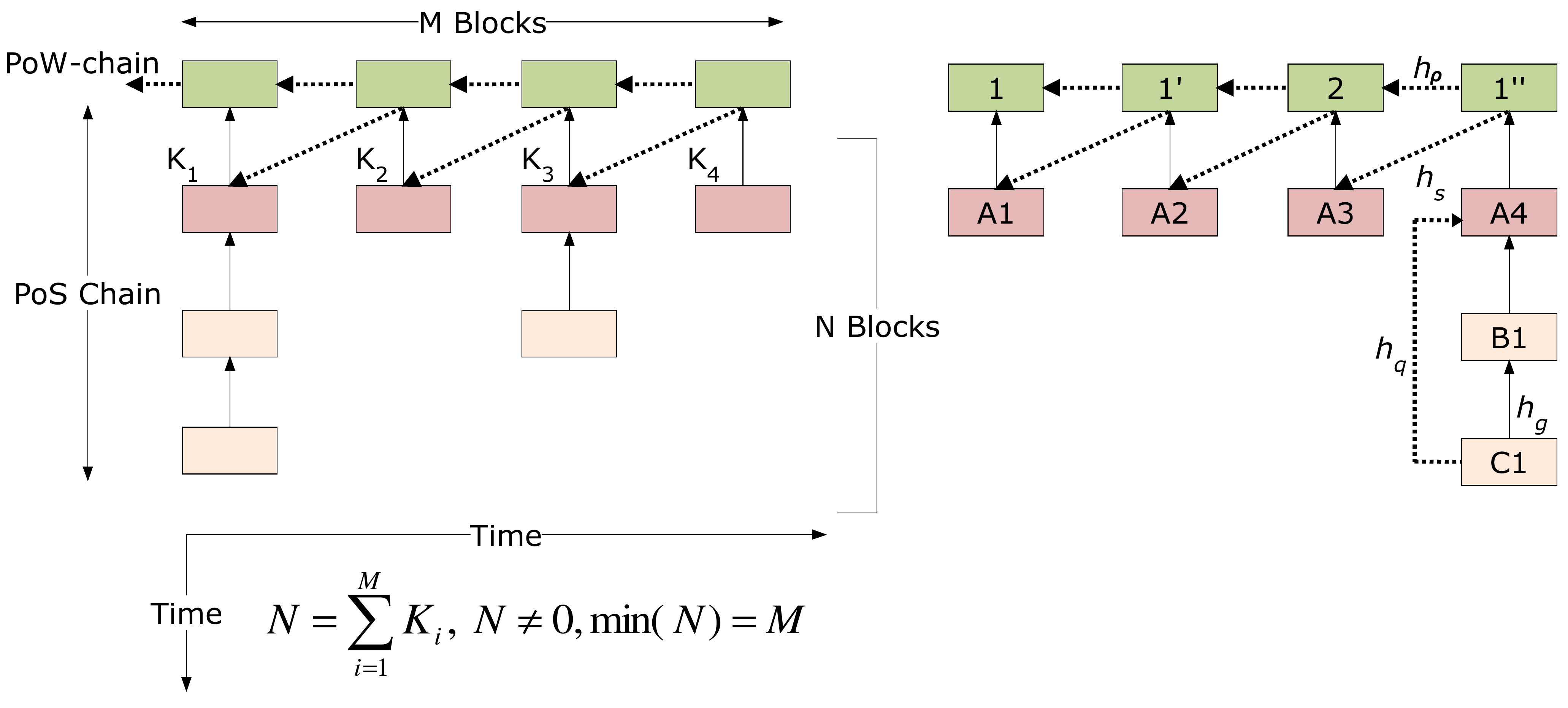}
  \caption{An overview of the proposed blockchain structure with growth in two directions.}\label{figure3}
\end{figure}
In \emph{Reinshard}, the initial blocks of the chain-pairs are generated using the principles of PoW in 2-hop blockchain~\cite{Duong2017}, which on lateral functioning are joined by many sub-PoS blocks. Thus, growing the chain in two directions. In general, the proposed blockchain structure can be expressed as follows:
\begin{definition}\label{definition_block}
The dual-blockchain structure can be defined as a variant of the 2-hop chain-pair denoted by $B_{G}=\langle C_{W}, C_{S} \rangle$, where $C_{W}$ and $C_{S}$ are the set of PoW blocks (for miners) and PoS blocks (for stakeholders) with $M$ and $N$ blocks, respectively, and $B_{G}$ is the global blockchain. Here, each pair (by definition of 2-hop blockchain~\cite{Duong2017}) will be considered valid if every PoW block is possessing at least one PoS block.
\end{definition}
\begin{definition}
Each PoW-PoS chain-pair can further have pseudo blocks which are denoted by a set $C_{W}^{(P)}$, such that $C_{W}^{(P)}\subset C_{W}$ and $|C_{W}^{(P)}| < |C_{W}|$, where $|.|$ denotes the cardinality. As stated in Definition~\ref{definition_block}, $M$ also involves the pseudo blocks. The pseudo blocks are used to include additional PoS blocks under the same PoW-PoS chain-pairs respecting the limits of maximum PoS blocks in each sub-chain.
\end{definition}
\par
In the proposed model, $N=\sum_{i=1}^{M} K_{i}$, where $K_{i}$ denotes the number of sub-blocks with the $i$th PoW block including the one associated before the attachment of the incoming blocks. Hence, the minimum number of PoS blocks in the entire chain is given as $N_{min}=M-M'$, where $M'(=|C_{W}^{(P)}|)$ is the number of pseudo blocks. The key factor in the formation of this initial blockchain is the identification of the maximum permissible value of $K$ so that the performance, as well as the security, can be realized simultaneously.

\textit{How many PoS blocks can be appended to a valid PoW-PoS chain-pair?} This deals with the optimization aspect and should be adjusted dynamically, which means that the chain structure should be able to grow as well as shrink. This adds to the complexity of guessing the difficulty adjustment for an adversary and it cannot identify the leader amongst the valid chain-pairs; alongside, it fails to track the nodes' locations in the shards. The key factor to account for the chain length is the stake-ownership for the PoS-chain by the nodes with the added blocks. If $S_{max}$ and $S_{min}$ are the maximum and the minimum number of stakes, then the length of the PoS chain under each PoW block must satisfy the earned stakes, as $S_{min} \leq S_{i} \leq S_{max}$. However, an attacker may show high stakes and can affect the chain limits by controlling the chain length leading to a single-shard takeover attack. This can be prevented by optimal calculation of upper bounds of $K$, and even if an attacker controls the operations, it cannot guarantee the growth of the same PoS chain beyond a certain range. Such a mechanism will prevent false block generation as well as help other nodes to identify the faulty operations. 
\subsection{Reinforcement Learning (RL) based-bound}
RL-based bound is an important strategy to be used in the case of uncertain environments. RL helps to accumulate decisions based on certain rewards that are observed when the environment transits from its current state to the desired state. In \emph{Reinshard}, PoW-PoS chain-pairs (part of $B_{G}$) are the environments which are in a state possessing a certain number of PoS blocks and the next state refers to the inclusion of an incoming PoS block. In the proposed approach, the attachment through RL is based on metrics of the nodes owning a valid PoW-PoS chain-pair. The properties involve a list of metrics that are to be considered for deciding the rewards on transitions which help to fix the upper bounds on $K_{i}$. These metrics include, the number of permissible pseudo IDs ($\eta_{s}$), available storage ($G_{a}$), expected required storage ($G_{e}$), available stakes ($S_{a}$), Computation-power reward ($\omega$), and chain-growth rate reward ($\theta$). The computational power and block operations (chain growth rate) are presented as a value where the physical meaning is mapped to parameters, $\omega$ and $\theta$, respectively. All these metrics are operated over RL-model to help decide the rewards, which are translated into the upper limits of the number of blocks that must be present or allowed as a part of a sub-PoS chain. The use of RL-rewards helps to ensure that actual details on the configurations of nodes are never revealed to its peers, which prevents attackers to know the possibilities of incoming blocks. Moreover, it also prohibits the attacker to generate more PoS blocks to prevent general PoW-PoS operations. Such provisioning prohibits the long-range attacks as well as prevents an excessive number of untraced pseudo-IDs which may otherwise increase the performance overheads. Considering this, the additional number of blocks that can be accommodated under one pseudo ID of a valid chain-pair is given by:
    \begin{equation}\label{eq:4}
        K_{i,t}^{(A)}=\frac{G_{a}}{\eta_{s}G_{e}},
    \end{equation}
provided that $G_{e}\eta_{s} \leq G_{a}$. If this condition is unsatisfied, the new blocks cannot be included under the valid chain-pair. Consequently, the permissible range of $K_{i}$ at $t+1$ becomes $K_{i,t}+K_{i,t}^{(A)}$. On observing this value, each pseudo ID group of the valid PoW-PoS chain-pair is evaluated for their chain length, and the one with available slot (position for the block) and being a leader can accommodate the incoming block. In case the available space is split between the two or more pseudo chains, the attachment is unaffected as one block is appended at a time. In case the blocks are in batches, the First Come First Serve (FCFS) is considered, and in the case of the parallel instances, the blocks are assigned in order of their memory utilization (decreasing order). There are certain cases where the number of pseudo IDs possesses similar properties (especially when the blockchain is initiated or additional nodes join), then the selection of pseudo chain is done using $S_{a}$, $\omega$ and $\theta$. As a solution, the pseudo chain with the corresponding users having $\max(S_{a})$ is selected to accommodate the new PoS block. Scenarios, where the stakes cannot be distinguished, account for reward values of $\omega$ and $\theta$ to select a sub-chain with the maximum reward-value. The reward-value based on $\omega$ and $\theta$ is calculated as:
    \begin{equation}\label{eq:5}
        R(\omega,\theta)_{i,(t)}=\theta_{i,(t)}^{-1}+\gamma_{(t)}\;\omega_{i,(t)}^{-1},
    \end{equation}
where $\gamma$ is the controlling parameter which shows the association and is depicted as the probability of successful identification of a valid chain-pair. For most of the cases, $\gamma=1$, as the ledgers are always aware of the associated blocks. Here, $\omega$ is evaluated as the maximum likelihood estimate for the mining-rewards earned (upon validation), for the duration $t$, considering that the block operations follow a single parameter beta distribution~\cite{gupta2004handbook}, such that
    \begin{equation}\label{eq:6}
        \omega_{i, (t)}=K_{i,(t)}\left({\sum_{j=1}^{K_{i,(t)}}\log\left(\frac{1}{1-\left(\frac{O_{a,(t)}}{O_{e,(t)}}\right)_{j}}\right)}\right)^{-1},
    \end{equation}
where $O_{a}$ and $O_{e}$ are the actual earning rate and expected earning rate (these metrics are affected by the difference in the hash rate and the difficulty of mining), respectively. This helps to prevent the dominance effect\footnote{Consider an application process, where a server keeps on earning mining-rewards and manipulates users. In such a case, the server imbalances the PoW-PoS chain-pairs and leads to a dominance effect. This will lead to several attacks that are bounded to occur in an uncontrolled PoS chain.}. However, in non-available situations, the reward is calculated based on $\theta$, which is calculated as:
    \begin{equation}\label{eq:7}
        \theta_{i,(t)}=\left(\theta_{i,(t-1)}\left(1+\delta_{(t-1)}\right)^{t}\right),
    \end{equation}
where $\delta$ is the chain growth rate, which is calculated as the ratio of the actual number of blocks appended to the PoS-sub-chain to the expected number of appends. In a fully-aware scenario, this ratio is taken to the total PoS-blocks in all the valid chain-pairs. It is to be noted in (\ref{eq:5}) that the chain with subsequent lower values for the observation rates defined for incoming rewards and the chain growth is given preference in selection. This prevents long-waits as well as increases participation. In the scenarios, where the above designing does not account for the inclusion of a new PoS block, a valid (\textit{winner}) PoW-PoS chain-pair may delegate the incoming PoS block to another valid PoW-PoS chain-pairs. 
The procedures for handling incoming PoS block by the winning chain-pair are given in Algorithm~\ref{algo_rl}.
\begin{algorithm}[!ht]
\fontsize{8}{10}\selectfont
\SetAlgoLined
\KwResult{PoS block allocation}
 $V$=$\langle C_{W}, C_{S} \rangle^{V}$\Comment*[r]{listing valid chain-pairs}
 obtain (at $t$) $\eta_{s}$, $G_{a}$, $G_{e}$, $S_{a}$, $\omega$, $\theta$;\\
 $J_{V}$=$check\_validity(V,bool)$;\\
 \If{approach==reinforced}{
        \While{$pos\_decision$!=True}{
             select $i$=$valid(J_{V})$\Comment*[r]{winner chain-pair}
             $\eta_{x}$ = $count(pseudo\_ids\leftarrow valid(J_{V}))$;\\
             $K_{i}^{(A)}=\frac{G_{a}}{\eta_{x}G_{e}}$\Comment*[r]{extra block accommodations}
             $K_{max}$=$K_{i}^{(A)}+strength(i)$\Comment*[r]{maximum PoS blocks}
              \eIf{$\eta_{x}\leq \eta_{s}$}{
              \uIf{$K_{i}^{(A)} > 1\;\; \&\& \;\;\eta_{x} > 1$}
                {
                 \While{$j\leq \eta_{x}$}{
                        \If{$K_{j}^{(A)}+strength(j)\leq K_{max}$}{
                            $X$=$\langle J_{V,j} \rangle$; \\
                        }
                    $j$=$j+1$\\
                    }
                 \eIf{$count(X)$==$1$}{
                            $allocate\_unobject\_reinforced(X)$; \\
                    }
                    {
                    \While{$q$ $>$ $count(X)$}{
                            $S_{a,q}$=$stakes(X_{q})$;\\
                            $q$=$q+1$;\\
                        }
                        \eIf{$Equal(S_{a})$==True}{
                            $Z$=$count(Equal(S_{a}))$\\
                            \While{$r \leq Z$}{$R(\omega,\theta)_{r}$=$\theta_{r}^{-1}+\gamma\;\omega_{r}^{-1}$;\\
                            $r$=$r$+1;\\
                            }
                            $Y$=$Z\leftarrow\max(R(\omega,\theta))$\Comment*[r]{best pair}
                            $allocate\_unobject\_reinforced(Y)$;\\
                        }
                        {
                            $Y$=$\max(S_{a})$\Comment*[r]{best pair}
                            $allocate\_unobject\_reinforced(Y)$;\\
                        }
                    }
                }
                \uElseIf{$K_{i}^{(A)}$==$1$ $||$ $\eta_{x}$==$1$}{
                    $allocate\_unobject\_reinforced(i)$;\\
                }
                \Else{
                    $wait\_or\_delegate()$;\\
                }
              }
               {
                $wait\_or\_delegate()$;
               }
             $Logs()$
        }
 }
 \caption{PoS block allocations with RL.}\label{algo_rl}
\end{algorithm}
To model delegation, the model uses Q-learning RL, such that $B_{G}$ forms the environment considering a given chain-pair, $\langle C_{W}, C_{S} \rangle^{V_{x}}$, where $V_{x}$ denotes the valid pair. Now the set of actions are denoted as $\langle C_{W}, C_{S} \rangle^{A}$, which includes two possible values, allow (delegate) and disallow. Using these, a timely Q-learning table is generated to understand the value of $K_{i}$ for each instance. Note that the storage of the Q-learning table is for the purpose of debugging and its visibility is subject to the configuration of the deploying application. For delegation, the PoW-PoS chain-pair are bid around the non-inverse maximization of (\ref{eq:5}), which means that the choice of next PoW-PoS chain-pair, by definition~\cite{melo2001convergence}, is based on:
    \begin{equation}\label{eq:8}
        \max(Q(\langle C_{W}, C_{S} \rangle^{V_{x}}_{t+1},\langle C_{W}, C_{S} \rangle^{A})),
    \end{equation}
such that
   \begin{eqnarray}\label{eq:9}
        Q(\langle C_{W}, C_{S} \rangle^{V_{x}}_{t},\langle C_{W}, C_{S} \rangle^{A})=\nonumber\\
        Q(\langle C_{W}, C_{S} \rangle^{V_{x}}_{t-1},\langle C_{W}, C_{S} \rangle^{A})+ \delta_{(t-1)}.\nonumber\\
        \left(R'(\omega,\theta)+P_{x}. \max(Q(\langle C_{W}, C_{S} \rangle^{V_{x}}_{t+1},\langle C_{W}, C_{S} \rangle^{A})))\right),\nonumber
    \end{eqnarray}
where
    \begin{equation}\label{eq:10}
        R'(\omega,\theta)_{i,(t)}=\theta_{i,(t)}+\gamma_{(t)}\;\omega_{i,(t)}.
    \end{equation}
Here, $\delta\neq 0$, $Q(\langle C_{W}, C_{S} \rangle^{V_{x}}_{0},\langle C_{W}, C_{S} \rangle^{A})=0$, and $P_{x}$ is the probability of reserving a future slot for the incoming blocks. As it is difficult to predict such reservations, we model this probability around controlling parameter, such that $P_{x}=\gamma$. The physical meaning of the maximization means that the delegated PoS block must not over consume the resources of the PoW-PoS chain-pair to which it is delegated leaving no space for its own PoS blocks. Thus, the chain-pairs which show sufficient intake possibilities are selected for delegation.

\subsection{Chain-growth and difficulty adjustments}\label{chain_growthsection}
The chain-growth in the proposed blockchain is observed in two directions, one for the PoW-PoS chain-pair and another for PoS-sub-chain.

\noindent\textbf{--PoW-PoS chain-pair:} The operations of the PoW-PoS chain-pair in the proposed model are inspired by 2-hop blockchain~\cite{Duong2017}. However, the PoW blocks which are ignored from chain inclusion have to re-approve their mechanism and convert itself to a valid pair by first joining through a PoS sub-chain. The block format in the chain-pair is similar to 2-hop blockchain (Fig.~\ref{figure3}) as $H(h_{\rho}, h_{s}, \rho) < T$. However, the expansion of the PoW-PoS chain-pair is controlled by the current pseudo-pair generation rate of the entire blockchain. This means, if $\alpha_{m}$ is the pseudo-pair generation rate, then the new chain-pair can be accommodated according to,
\begin{equation}\label{eq:11}
   H\left(\alpha_{m}, Z\left(h_{\rho}, h_{s}\right), \rho\right) < T,
\end{equation}
where $H(.)$ and $Z(.)$ are the respective cryptographic hash functions. The proposed model uses pseudo-blocks' generation rate as it focuses on improved applications along with sharding. The function uses parts of the blocks to which the newly added chain-pair must be appended to provide reliable as well as sustainable growth of the chain. This chain pairing helps to prevent the major of the chain-controlling attacks, as well as it increases the participation by allowing the blockchain nodes to control the pseudo generation rates. Additionally, it prevents multiple forks and single node control over the incoming PoS blocks.\\
\noindent\textbf{--PoS sub-chain:} This chain operates in three parts. At first, a list of valid chain-pairs is selected, who can compete to be the leader, and is selected using the concept of a VDF (construction provided in Appendix.A of the supplementary file), and finally, the PoS block is ready to be allocated to a chain-pair. VDF prevents PoS blocks from attempting to attach to additional chain-pairs at the same instance because of its core ideology of sequential work with deterministic functions. From the application's perspective, VDF helps to delegate in situations of unavailable miners or stakeholders as well as model the wait algorithms.\\
\noindent$\bullet$ \textbf{List of validators.} This list includes the valid chain-pairs who can compete to be the leader for extending the PoS sub chain under their pre-attached PoS block. The validators are selected based on the original rewards earned as a part of the RL mechanism, such that all the chain-pairs, which holds $R'(\omega,\theta)_{i,(t)} \geq \textnormal{ mean}(R'(\omega,\theta)_{\sum_{i},(t-1)})$ are marked as valid. Additionally, the valid chair pairs must satisfy the requirements in (\ref{eq:4}). Furthermore, the valid chain-pair must hold $\eta_{x}\leq \eta_{s}$ and a valid $\alpha_{m}$ with a consideration that each pair must be complete and no PoW-PoS chain-pair with a pending status can be included in the list of the validators.\\
\noindent$\bullet$ \textbf{Selection of a leader.} A VDF is defined as a function $\mathcal{W}$=(\textit{Setup}, \textit{Eval}, \textit{Verify}), considering that the operations of \textit{Verify} take much lesser time than the \textit{Eval}~\cite{Boneh2018}. This means that the node with a valid PoW-PoS chain-pair must possess sufficient computational ability than the nodes of the corresponding sub-chain, which will prevent an adversary from taking control by solving the evaluation puzzle in extremely lesser time than the verification puzzle. The solution is modeled around the time ($\tau$), which is taken to evaluate a valid pseudo-chain-pair. In the given model, out of the validators, the one with a valid $(\mathcal{E}_{k}, v_{k})$ is selected as a leader.
\begin{enumerate}
\item \textbf{\textit{Setup}:} The proposed model operates with the basic assumption of a VDF for generating the PoS block, i.e, the public parameters, $(\mathcal{E}_{k}, v_{k})$, are available for evaluating and verifying the blocks within the polynomial limits (sub-exponential time as stated in Boneh et al.~\cite{Boneh2018}). With such public parameters, the validator, which solves
        \begin{equation}\label{eq:11b}
             \widetilde{H}\left(B_{PoW}, v_{k}\right)<R'(\omega,\theta)\;\widetilde{T},
        \end{equation}
        is elected as the leader. Accordingly, the public parameters are obtained through the randomized \textit{Setup} algorithm which uses $(\lambda, \zeta)$ as its parameters with $\zeta$=$f(\tau, \widetilde{T})$.
\item \textbf{\textit{Evaluation}:} The leader performs the evaluation function by considering the allocation mechanism which will be used for placing the block under the required chain-pair. For this, the leader uses the hash of the parent chain-pair's PoS block ($h_{q}$) and the hash of the previous PoS block ($h_{g}$). The \textit{Eval} algorithm takes $\mathcal{E}_{k}$ and $\mathcal{I}=H(h_{q},h_{g})$, where $\mathcal{I} \in \{0,1\}^{\lambda}$ operates with $\tau$ which is the current time consumption as per $\alpha_{m}$, to generate an output $\mathcal{O} \in \{0,1\}^{\lambda}$ as an image of $\mathcal{I}$ along with a proof $\pi$.
\item \textbf{\textit{Verify}:} Once the \textit{Eval} is executed, the PoS block is generated using $B_{PoW}$, $v_{k}$, $\mathcal{I}$, $\mathcal{O}$, and $\pi$. The \textit{Verify} function is evaluated with ($v_{k}$, $\mathcal{I}$, $\mathcal{O}$, $\pi$) and a Boolean is recorded as an observation. Upon success, the PoS block is appended to the intended location which was used by the leader during the evaluation procedure. Any adversary, which poses to be a valid PoW-PoS chain-pair, needs to solve the \textit{Eval} function in a shorter time than that of \textit{Verify}. As the allocation is decided by the winning pair, an adversary needs to acquire all the possible locations as well as run the \textit{Eval} for all the cases under time extremely lower than $\tau$, which is computationally expensive even with influential hardware. Thus, the security of the PoS block is supported by the choice of allocation procedure (known locally only to the leader) and the time-difficulty of the VDF.
\end{enumerate}
\begin{figure*}[!ht]
  \centering
  \includegraphics[width=0.7\linewidth]{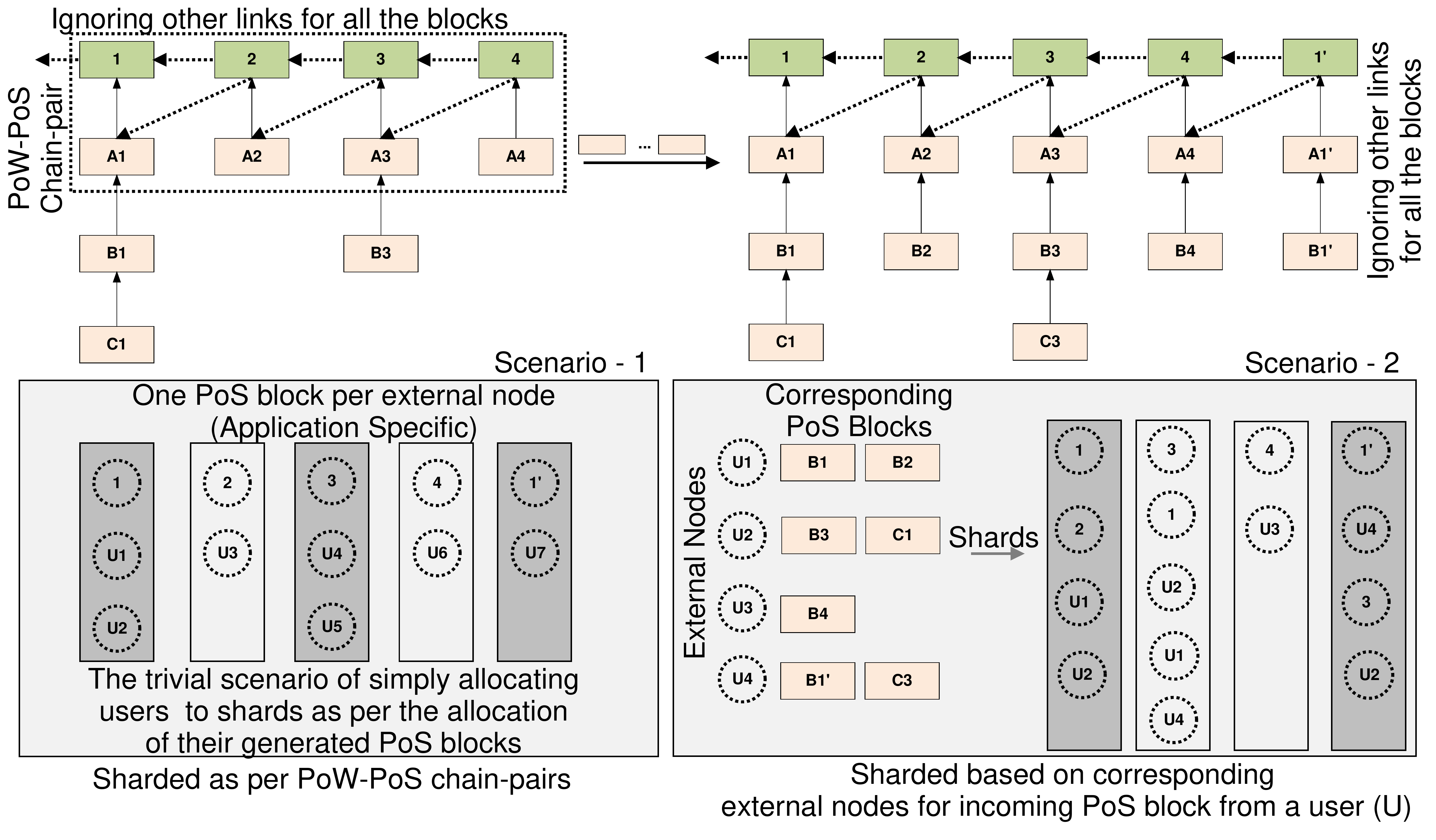}
  \caption{An overview of the sharding process (storage to processing).}\label{Figure5}
\end{figure*}

\noindent\textbf{--Difficulty adjustments:} The primary target of the proposed model is to develop a generic system that can be widely adopted for the major of the applications listed for the blockchain. The difficulty adjustment is required to control the extreme growth as a valid chain-pair that may show excessive storage to keep increasing its pseudo pairs or sub-chains. Thus, in the main PoW-PoS chain-pair, the difficulty is adjusted as:
\begin{equation}\label{eq:12}
T_{r+1}=\left(\frac{\max(\eta_{x,\langle C_{W},C_{S}\rangle^{V_{i}}})}{\sum\limits_{i\neq j}(\eta_{x})_{j}}.\frac{R'(\omega,\theta)_{i,\langle C_{W},C_{S}\rangle^{V_{i}}}}{\sum\limits_{i \neq j} R'(\omega,\theta)_{j,\langle C_{W},C_{S}\rangle^{V_{j}}}}\right)T_{r},
\end{equation}
where $T_{r+1}$ is the next difficulty level w.r.t. time $t$, and (\ref{eq:12}) controls the winning pair and prevents reducing the participation and control over the blockchain. Although the PoW-PoS chain-pair difficulty and time delay from VDF are enough to control the global blockchain, to manage the deadlocks/concurrent operations and to control the pseudo-pair competition, difficulty adjustment is also modeled for sub-chain, such that, for $K^{(R)}$ number of incoming PoS blocks,
\begin{equation}\label{eq:13}
\fontsize{8}{10}\selectfont
\widetilde{T}_{r+1}=\begin{cases}
\left(\frac{K^{(A)}_{i,\langle C_{W},C_{S}\rangle^{V_{i}}}}{\sum\limits_{i\neq j}(K^{(A)}_{x})_{j}}.\frac{S_{i,\langle C_{W},C_{S}\rangle^{V_{i}}}}{\sum\limits_{i \neq j} S_{j,\langle C_{W},C_{S}\rangle^{V_{j}}}}\right)\widetilde{T}_{r}, \sum\limits_{i\neq j}(K^{(A)}_{x})_{j} \geq K^{(R)}, \\
\left(\left(\frac{K^{(A)}_{i,\langle C_{W},C_{S}\rangle^{V_{i}}}}{\sum\limits_{i\neq j}(K^{(A)}_{x})_{j}}\right)^{-1}.\frac{S_{i,\langle C_{W},C_{S}\rangle^{V_{i}}}}{\sum\limits_{i \neq j} S_{j,\langle C_{W},C_{S}\rangle^{V_{j}}}}\right)\widetilde{T}_{r}, otherwise.
\end{cases}
\end{equation} 

\begin{algorithm}[!ht]
\fontsize{8}{10}\selectfont
\SetAlgoLined
\KwResult{Sharding and ledger division}
 Input: PoS\_block\_allocation()\\
 set $inst$=$node\_ability()$\Comment*[r]{1 or more PoS blocks}
 set $i$=0, $j$=0, $q$=0, $d$=0;\\
 obtain $N$=extern\_node\_value(PoS\_block\_allocation);\\
 \eIf{$inst$==$1$}{
    \While{$i$ $\leq$ $M$ \&\& $sum(j)$ $\leq$ $N$}{
    $c$=$strength\_real(i)$\Comment*[r]{Allocated Sub-PoS blocks}
        \While{$j$ $\leq$ $c$}{
            $shard[i]$=$store\_chain$($K_{j}$,$i$,$c$);
            $j$=$j$+$1$;\\
        }
    $j$=$j$+$1$;\\
    }
    $store\_advertise\_reward(shard)$\Comment*[r]{shard and reward}
 }
 {
 \While{$q$ $\leq$ $count(external\_nodes)$}{
        $ex\_loc$=$external\_block\_alloc()$\Comment*[r]{block location}
        \While{$d$ $\leq$ $count(ex\_loc)$}{
            $z$=$valid\_chain\_pairs(d,ex\_loc)$;\\
            $shard[q]$=$store\_chain$($K_{z}$,$q$,$z$);\\
            $d$=$d$+$1$;
            }
        $q$=$q$+$1$;\\
    }
    $store\_advertise\_reward(shard)$\Comment*[r]{shard and reward}
 }
 \caption{Auto-sharding with \textit{Reinshard}.}\label{algo5}
\end{algorithm}

\noindent\textbf{--Sharding (Storage and Processing):}
A fully-scaled ledger may lead to several overheads in terms of storage, synchronization, and sharing, whereas the partial-ledgers are better in terms of performance but require additional protocols for communications. Using external methods for sharding in a pre-formed blockchain leads to several overheads. This is overcome in the proposed model as the node allocations across the valid chain-pairs provide a unique auto-sharding feature which helps to lower the burden in terms of ledger maintenance as well as reduce the synchronization problem. The newly formed blockchain helps to attain sharding by architecture. An illustration of sharding can be observed in Fig.~\ref{Figure5} with procedures in Algorithm~\ref{algo5}. Here, function $node\_ability()$ calculates the ability of the node to have 1 or more PoS blocks. Function $extern\_node\_value(Pos\_block\_allocation)$ obtains the external node value for the allocated PoS block, function $strength\_real()$ calculates allocated sub-PoS blocks, function $store\_chain()$ stores the formed chain to shards, function $store\_advertise\_reward()$ rewards for the shard formation, function $count()$ calculates the number of nodes, function $external\_block\_alloc()$ expresses the block location and function $valid\_chain\_pairs$ validates the chain pairs. As shown in Fig.~\ref{Figure5}, \textit{Reinshard} can be operated in two ways, the first one involves a dedicated application-specific sharding in which each external node is entitled to generate a single PoS block (e.g. use for login) and the shards are formed based on the attachments of PoW-PoS chain-pairs. The second involves multiple PoS blocks from external nodes, which are allocated as per the decision of the valid PoW-PoS chain-pairs. In such a scenario, the shards are formed following the external node contents and the location of their PoS blocks. In the first case, each external node is entitled to a single stake each time a new PoS block is appended to their sub-chain and it is convenient to manage and control such a blockchain. However, unlike this, in the second case, the stakes are rewarded every time a new PoS block is appended as well as these are given stakes out of the rewards earned by the valid PoW-PoS chain-pair on their selection as a leader.
\begin{figure}[!ht]
  \centering
  \includegraphics[width=210px]{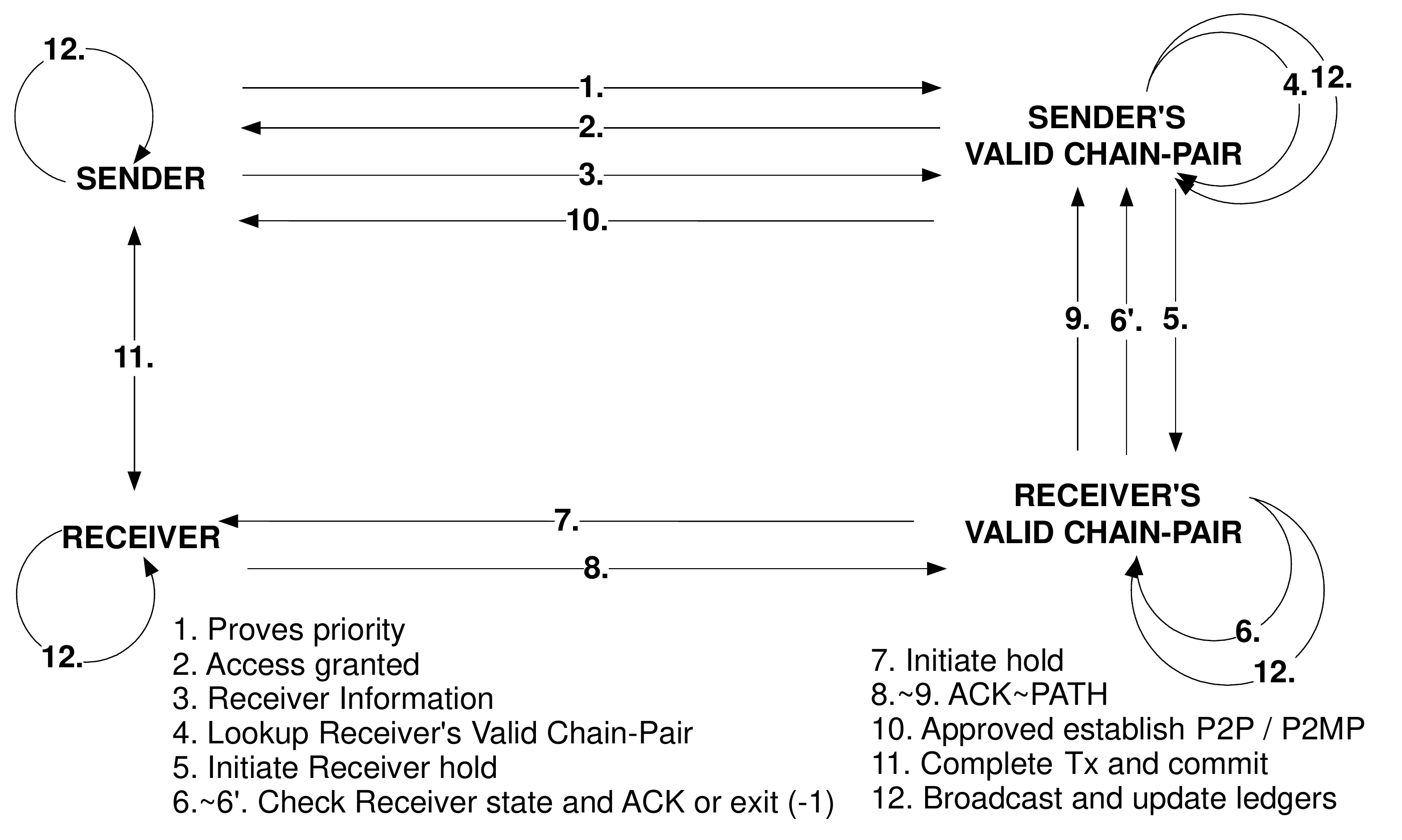}
  \caption{An overview of states for inter-shard communications.}\label{Figure6}
\end{figure}

\noindent\textbf{--Concurrency resolutions:} Concurrency is a correct measure of efficiency for sharded blockchains. In the proposed approach, at first, priority $\psi$ is defined to select the receiver, as shown in Fig.~\ref{Figure6}. This priority helps to control the concurrent operations as well as prevent any deadlocks during transactions. $\psi$ is calculated based on the rewards/stakes earned by a generator (external node). It is similar to the booking of the receiver node to share some information. However, with unlimited control over the receiver, the entire blockchain can undergo certain attacks (single-shard takeover) as it will prevent any sort of updates across the chain. Thus, to prevent this, a VDF mechanism is adopted and a sender with priority $\psi$ can only hold the receiver for $(\tau'.\mathcal{D}+\Gamma)$ duration, ($\tau' < \tau$), where $\mathcal{D}$ is the number of involved valid chain-pairs which are used to trace the receiver and $\Gamma$ is the network latency. Note that each node may not have an average time greater than $\tau$ as it may compute \textit{Eval} faster than \textit{Verify}. This allows simultaneous control over the intended receivers during concurrent operations. Once these controls are available, the transaction can be processed, and ledgers can be updated based on the sharded blockchain.

\noindent\textbf{--Intra and Inter-shard communications:} The intra- and inter-shard communications are bounded by the concerns of participation, which is easier when operating in a single shard and becomes complex when multiple shards, as well as multiple receivers, are involved. The two Algorithms~\ref{algo6} and~\ref{algo7} facilitate the required mode of transactions. The choice of P2P or P2MP depends on the number of participants as receivers and the selection of the protocol for communication depends on the underlying network architecture. The inter-shard communications proposed in this article uses the hold time which is derived based on the delay factor of the VDF, and it provides an effective strategy to resolve concurrency issues. One of the examples can be the \textit{Train}-and-\textit{Hotel} booking problem~\cite{Tribble}, where two transactions are involved and synchronous mode may lead to the non-availability of one of the two requirements (hotel or a train). To control this, the proposed algorithm initiates the hold procedures for the time which is enough to prevent an adversary from performing the \textit{Eval} function and try to manipulate the block before it is verified. The hold procedures bound the receivers until the connection (tunnel) is not formed between the receivers. With such facilitation, receivers of the different shards can be accommodated to wait for a duration, which is enough to bring participants to wait until the completion of one full transaction (booking of train and hotel). However, network latency plays a crucial role and its derivation and conceptualization are beyond the scope of this article.

\begin{algorithm}[!ht]
\fontsize{8}{10}\selectfont
\SetAlgoLined
\KwResult{Intra-shard communications}
 Input: shards, sender, receivers \\
 \eIf{$valid(sender,receiver())$==$true$}{
    $fetch\_info$($shards$)\Comment*[r]{sender's call to proceed}
    set P2P or P2MP\Comment*[r]{communication modes}
    initiate $protocol$\Comment*[r]{network layer operations}
    initiate $transactions$\Comment*[r]{Begin communication}
 }
 {
 exit(-1) \Comment*[r]{invoke control procedures}
 }
 $update\_ledgers()$\Comment*[r]{ledger update and rewarding}
 \caption{Intra-shard communications.}\label{algo6}
\end{algorithm}

 \begin{algorithm}[!ht]
\fontsize{8}{10}\selectfont
\SetAlgoLined
\KwResult{Inter-shard communications}
 Input: shards, sender, receivers \\
 Set timer;\\
 \eIf{$valid(sender,receiver())$==$true$}{
    $prove\_priority$($\psi$)\Comment*[r]{Fig.~\ref{Figure6}: step 1}
    $fetch\_info$($shards$)\Comment*[r]{Fig.~\ref{Figure6}: steps 2$\sim$3}
    $rec\_info(cp\_lookup(receivers))$\Comment*[r]{Fig.~\ref{Figure6}: step 4}
    $init\_receiver\_hold(\tau.\mathcal{I}+\Gamma)$\Comment*[r]{Fig.~\ref{Figure6}: steps 5$\sim$7}
    \eIf{valid($timer$,$\tau.\mathcal{I}+\Gamma$)==$true$}{
        set P2P or P2MP\Comment*[r]{Fig.~\ref{Figure6}: steps 8$\sim$10}
        initiate $protocol$\Comment*[r]{network layer operations}
    }
    {
    timeout(-1);\\
    }
    initiate $transactions$\Comment*[r]{Fig.~\ref{Figure6}: step 11}
 }
 {
 exit(-1) \Comment*[r]{invoke control procedures}
 }
   $update\_ledgers()$\Comment*[r]{Fig.~\ref{Figure6}: step 12}
 \caption{Inter-shard communications (Fig.~\ref{Figure6}).}\label{algo7}
\end{algorithm}


\section{Security evaluations}
The Reinshard chains are analyzed for security properties, such as chain growth, chain quality, and common prefix similar to ~\cite{kiayias2017ouroboros,david2018ouroboros,Duong2018,Duong2017,Hcashteam2017}.To understand the security evaluations, a block's position in the chain is referred to as its height. Furthermore, in each round, at least one block will be appended to the global chain or sub-chain using (\ref{eq:11}) or (\ref{eq:11b}). In \textit{Reinshard}, the entire blockchain is divided into two distinct chains, one is the subsidiary (sub-chain) chain (i.e., PoS-consensus), another is the global chain (i.e., hybrid (PoW + PoS) consensus) that contains PoW and PoS blocks. Notably, the global chain is similar to the 2-hop chain~\cite{Duong2017} (or TwinsCoin~\cite{Duong2018}). \\
\noindent\textbf{--Security analysis of \textit{Reinshard}-global Chain}. A chain-pair of \textit{Reinshard}-global chain, $\langle C_W, C_S\rangle$, is considered for evaluations. In order to extend the pair of \textit{Reinshard}-global blockchain, a PoW-miner needs to generate a PoW-block first, and then corresponding PoS-holders (or leader) will sign this block and generate a new PoS-block. Notably, both PoW-miners and PoS-stakeholders can be honest or malicious. Thus, the three properties are guaranteed by the following cases:\\
\noindent $\bullet$ The ideal case is both PoW-miners and PoS-stakeholders are honest which can guarantee the property of chain growth. The main reason is that malicious players cannot prevent the operations of the ideal case.\\
\noindent $\bullet$ The common case contains two distinct types. In the common case, the PoW block mined by the PoW-miner is corresponding to the PoS block and mapped to the PoS-stakeholder. The first type is that PoW-miner is malicious and PoS-holder is honest, and the honest PoS-holder will either sign the block mined by the malicious PoW-miner or discard it. Another type is that the PoW-miner is honest but the PoS-holder is malicious, here, the malicious PoS-holder will either sign the block mined by the honest PoW-miner or discard it. Notably, if the probability of the common case is smaller than the ideal case, then malicious players cannot generate more PoS-blocks than honest players. Even if they win all the competitions, there are still some blocks remaining from honest players. Apparently, the chain quality can be guaranteed by the common case.\\
\noindent $\bullet$ The worst case is both PoW-miners and PoS-stakeholders are malicious. It is assumed that the probability to find a new PoW-block by all the PoW-miners in one round is very small. Following this, if all of the honest players do not receive the new block from some rounds, they would obtain the same best chain-pair $\langle C_W, C_S \rangle$. The reason is that, in the worst case, all the honest players have the same view of the global chain-pair. The common prefix property can be guaranteed due to the reason that malicious players do not have enough resources to corrupt and diverge the view of the honest players by sending new blocks regularly. The actual architecture itself is able to secure the entire chain against known attacks as the direct inclusion of the blocks is not possible and has to be earned based on sufficient storage and computational powers, which prohibits an adversary to be a part of the global chain. However, PoS sub-chains are not aloof from such conditions as an intermittent chain-pair holder (node) may go rogue and create multiple forks by generating as many false blocks to the sub-chains, which may complete the limits on the blocks and result in a deadlock. Thus, the security of the sub-chain is required to prevent such attacks.\\
\noindent\textbf{--Security analysis of \textit{Reinshard} sub-chain}.
In essence, this sub-chain is similar to the conventional PoS chain. The key point is that the sub-chain is realized via a VDF with a time delay parameter $\tau$. To obtain a secure sub-chain, \textit{Reinshard} needs to achieve the property of, in particular, persistence and liveness. In fact, persistence follows from the properties of a common prefix (or chain consistency) and chain growth, and liveness follows from the properties of chain quality and chain growth. Intuitively, the chain growth property states that the chains of honest players should grow linearly to the number of rounds. Meanwhile, because of the use of VDF, it is required that the verification time $t_V$ is less than the evaluation time $t_E$ ($=\tau$), where $t_V < t < t_E$, which implies that the existing participants will be verifying the signed blocks in time lesser than that required by an intruder to generate the new block, which is equivalent to or greater than the evaluation time. The following properties from the existing approaches~\cite{kiayias2017ouroboros,Duong2017,Duong2018,garay2015bitcoin} help to formally understand the correct functioning of the proposed blockchain.
\begin{definition}[Chain Growth]
For all shards, each honest player finalizes chain-pairs at the end of their round $\roundr$ and has length $\mathcal{L}\geq \mathcal{L}-\mathcal{K}$ for the growth parameter $\mathcal{K}$, such that $\mathcal{L} \leq K_{max}$. Following this, each honest chain-pair must have a synchronized value of $N$.
\end{definition}
\begin{lemma}[Chain Growth Lemma~\cite{garay2015bitcoin}]
If an honest party has a chain-pairs with length $\chainlength$ at the round $\roundr$, then every honest party has adopted chain-pairs of length at least $\left(\chainlength+\sum^{\rounds-1}_{i=\roundr}K_i\right)$ ($\leq K_{max}$) by the round $\rounds\geq \roundr$ and $N$ for all parties must be same when observed from the last appended block.
\end{lemma}
\begin{proof}
By induction on $\rounds-\roundr\geq 0$, and assuming the basis $(\rounds=\roundr)$, if an honest party has a chain $\chainC$ with length $\chainlength$ at the round $\roundr$, then the party broadcasts $\chainC$ at a round earlier than $\roundr$. It follows that every honest party will receive $\chainC$ by the round $\roundr$. For the inductive step, according to the inductive hypothesis, every honest party has received a chain of length at least $\chainlength'=\chainlength+\sum_{i=\roundr}^{\rounds-2} K_i$ by the round $\rounds-1$. Obliviously, in this setting, if $K_{\rounds-1} = 0$ the statement follows directly, so assume $K_{\rounds-1}= 1$. Notably, $K_i$ implies that the expectation of the block number mined by an honest player after $i$ rounds. Furthermore, it is to be noted that every honest party can query the valid chain-pair with a chain of length at least $\chainlength'$ and $N'$ at the round $\rounds-1$. It follows that all honest parties successful at the round $\rounds-1$ broadcast a chain of length at least $\chainlength'+1$ and sum $N$.
Since $\chainlength'+1=\sum_{i=\roundr}^{\rounds-1}K_i$, and $N'$ is the same for all, it completes the proof.
\end{proof}
The chain quality property guarantees that there will eventually be a block in the finalized chain-pair that was proposed by an honest player subject to the limits imposed by $K_{max}$. In other words, the property of chain quality aims at expressing the number of honest blocks' contributions that are contained in a sufficiently long and continuous part of an honest chain. Here, $\Upsilon$ is used to define the stakes ratio of the adversaries, $\varrho$ is used to define the stakes ratio of the honest holds and $\epsilon$ acts as the system parameter, and $\epsilon\in \{0,1\}$.
\begin{definition}[Chain Quality~$\Cq$~\cite{garay2015bitcoin}]
The chain quality $\Cq$ with parameters $\varrho \in R$ and $\chainlength\in \mathcal{N}$ state that for any honest party $P$ with chain $\chainC$, it holds that for any $\chainlength$ consecutive blocks of $\chainC$ the ratio of honest blocks is at least $\varrho$.
\end{definition}
\begin{theorem}[Chain Quality~\cite{garay2015bitcoin}]
Let $\astkratio-\sysp$ be the adversarial stake ratio. The protocol satisfies the chain quality property with parameters $\varrho\cdot (\astkratio-\sysp)=\astkratio/(1-\astkratio)$ and $\chainlength\in \mathcal{N}$ through an epoch of $R$ slots with probability at least
\begin{equation*}
  1-\exp(-f(\sysp^2\cdot (\astkratio\cdot\chainlength))+\ln R),
\end{equation*}
s.t.
\begin{equation}
    K_{i}^{(A)} \leq K_{max}-\chainlength.
\end{equation}
\end{theorem}
\begin{proof}
From the proof of chain growth, it is known that with high probability a segment of $\chainlength$ will involve at least $(1-\astkratio)\cdot\chainlength$ slots with honest leaders; hence the resulting chain must advance by at least $(1-\astkratio)\cdot\chainlength$ blocks, which is $\leq$ $K_{i}^{(A)}$. Similarly, the adversarial parties are associated with no more than $\astkratio\cdot \chainlength$ slots, and thus can contribute no more than $\astkratio\cdot \chainlength$ blocks to any particular chain over this period. It follows that the associated chain possessed by any honest party contains a fraction $\astkratio/(1-\astkratio)$ of adversarial blocks with probability $1-\exp(-f(\sysp^2\cdot (\astkratio\cdot\chainlength))+\ln R)$.
\end{proof}
\begin{definition}[Common Prefix~\cite{kiayias2017ouroboros,garay2015bitcoin}]
The common prefix (or chain consistency) implies that if $\chainC$ and $\chainC'$ are the finalized chains of two honest players, then $\chainC$ is a prefix of $\chainC'$ or vice versa at any point of time.
\end{definition}
\begin{lemma}[Common Prefix Lemma~\cite{kiayias2017ouroboros,garay2015bitcoin}]
If $\chainC_{1}$ is adopted by an honest party at round ${\roundr}_{,1}$, and $\chainC_{2}$ is either adopted or diffused (broadcast) by an honest party at round $\roundr$ and has $\strlen(\chainC_{2})\leq \strlen(\chainC_{1})$, then $\chainC_{1}$ is a prefix of $\chainC_{2}$ or vice versa at any point of time for consecutive rounds.
\end{lemma}
\begin{proof}
In the sharded blockchain, if the honest players receive different chains for different intervals, both the chains, i.e., $\chainC_{1}$ and $\chainC_{2}$ are the prefix of a common chain that proves the Lemma. Additionally, the validation of PoS sub-chain can be guaranteed in any $\roundr$ following the honesty of the valid chain-pair, which ensures the correctness of the sub-chains.
\end{proof}
\begin{definition}[Chain Wait]
For the sharded blockchain, this property ensures that only valid chain-pair generates the blocks (equal opportunity to all participants) for the inter-shard communications, and minimum wait for the concurrent operations has been followed. This is guaranteed by the fact that each valid party must advertise $t_V$ and $\forall$ $\roundr$, $t_V < t_{E}$, which is known to all the chain-pairs.
\end{definition}
\begin{lemma}[Chain Wait Lemma]
For a given inter-shard communications, if two different $\tau_{1}$ and $\tau_{2}$ are observed from valid chain-pair and the intended receiver by the sender (initiator), both $\tau_{1}$ and $\tau_{2}$ are equal in a valid blockchain and  $t_V < (\tau_{1}=\tau_{2})  <t_{E}$.
\end{lemma}
\begin{proof}
Consider a scenario where a sender has initiated a request for two different nodes having locations in either the same or different shards. Now, the wait request for each query is a time $t_{1}$, which can be either decided by the sender or its corresponding chain-pair (depending on the mode of deployment and configurations). Now, a wait time $t_{2}$ corresponds to the intended receivers, and then a P2P or P2MP connection is initiated. The block signing is accompanied by advertising the $t_{2}$. For a valid chain, both $t$ and $t'$ are same and must be following   $t_V < (\tau_{1}=\tau_{2})  <t_{E}$. This ensures equal waiting for all the involved nodes and prevents intentional termination of connections when inter-shard communications are involved.
\end{proof}
\begin{definition}[Unbiased Sharding]
We say this is an \textit{Unbiased Sharding} in the sharded blockchain, if it satisfying the following requirements: 1). the process of generating the shards should be unpredictable and must not be controlled (or manipulated) by any single node, 2). the knowledge of nodes in the shards should not be predicted.
\end{definition}

\begin{lemma}[Unbiased-Sharding Lemma]
If an adversary $\mathcal{A}$ becomes part of the chain as a chain-pair, $\langle C_{W}, C_{S}\rangle$, or sub-chain $\langle C_{S}\rangle$, then no adversary can decide the shard-participants with an overwhelming probability.
\end{lemma}
\begin{proof}
In the sharded blockchain, the architecture is built such that it may grow or shrink for sub-chain based on the dynamic difficulty adjustment, which is unpredictable as stated in the chain-extension. The bounds on $K$ are governed by the RL-rewards and must be validated by the stakeholders. Thus, no chain-pair can extend it without validation. Once these are made, the state of the blockchain can be known to every participant. Now, the sharding is carried based on the location of the PoS block generated by a node. The node can generate any blocks, but the attachment is controlled by the winning chain-pair and verified by the valid chain-pairs. This means to bias the sharding, the adversary ($\mathcal{A}$) must be able to affect the verification and control the decision on sharding which is against the working of the VDF as the sub-node cannot present the computational requirement of being a part of the PoW-PoS chain pair. Even if it manages to show the same level of computational power, it has to solve the VDF puzzle under the verification time ($t_{V}$), which is practically impossible due to sequential steps in evaluation (governed by $t_{E}$) which have to be unique and need to be publicly verified. Thus, unbiasedness is guaranteed through the procedures of VDF used for extending the PoW and the PoS chains.
\end{proof}
\begin{remark}[Chain Availability]
In the case of \textit{Reinshard}, chain availability refers to the all-time accessibility of the node information despite the occurrence of failures in the targeted applications. It also includes the possession of information of sub-chain of the chain-pairs which have failed or inaccessible in any round.
\end{remark}
\begin{remark}[Non-cascading failures]
For general failures, the inclusion of the pointer to the PoS sub-block of the failed chain-pair can help to recover the entire sub-chain. However, in practical situations, a single server is the chain-pair generator, thus, the recovery is based on the condition that sharding information is available to all the neighboring servers. Additionally, in \textit{Reinshard}, the sharding helps to maintain the recovery of non-cascading failures by allowing sharded nodes to provide the information for the lost or unavailable nodes or even chain-pairs. Scenarios, where external nodes generate more blocks and lead to different shards, have better accessibility in the case of failures than the scenarios with only one block per node.
\end{remark}
\begin{remark}[Cascading failures]
In the case of cascading failures, even the nodes with a different number of blocks may fall short of recovering as consecutive failures affect the recovery. However, with \textit{Reinshard}, the scenario with multiple blocks is sharded as per the external nodes, which helps to retain maximum information and allow better recovery. At present, \textit{Reinshard} is able to detect failures with pseudo-chain-pairs and the major of other operations are evaluated by assuming that the generators of the chain-pairs are always available and accessible.
\end{remark}
\section{Performance evaluation and comparisons}
\textit{Reinshard} was evaluated on Intel\textsuperscript{\textregistered} Core\texttrademark i7-8750H CPU $@$ 2.20 GHz on a Dell G7 series workstations using instances from NoobChain~\cite{noobchain} with concurrent clients coded in Node.js\textsuperscript{\textregistered} which allows visualization through Chrome's JavaScript engine. The evaluations are presented in two parts, the first part helps to understand the run time operations of the proposed blockchain especially for sharding. The second part discusses the importance of using VDF by considering an exemplary concurrency problem (\textit{Train}-and-\textit{Hotel} booking problem~\cite{Tribble}).
\begin{figure}[!ht]
    \centering
   \begin{minipage}[b]{.5\linewidth}
        \centering
        \includegraphics[width=1\linewidth]{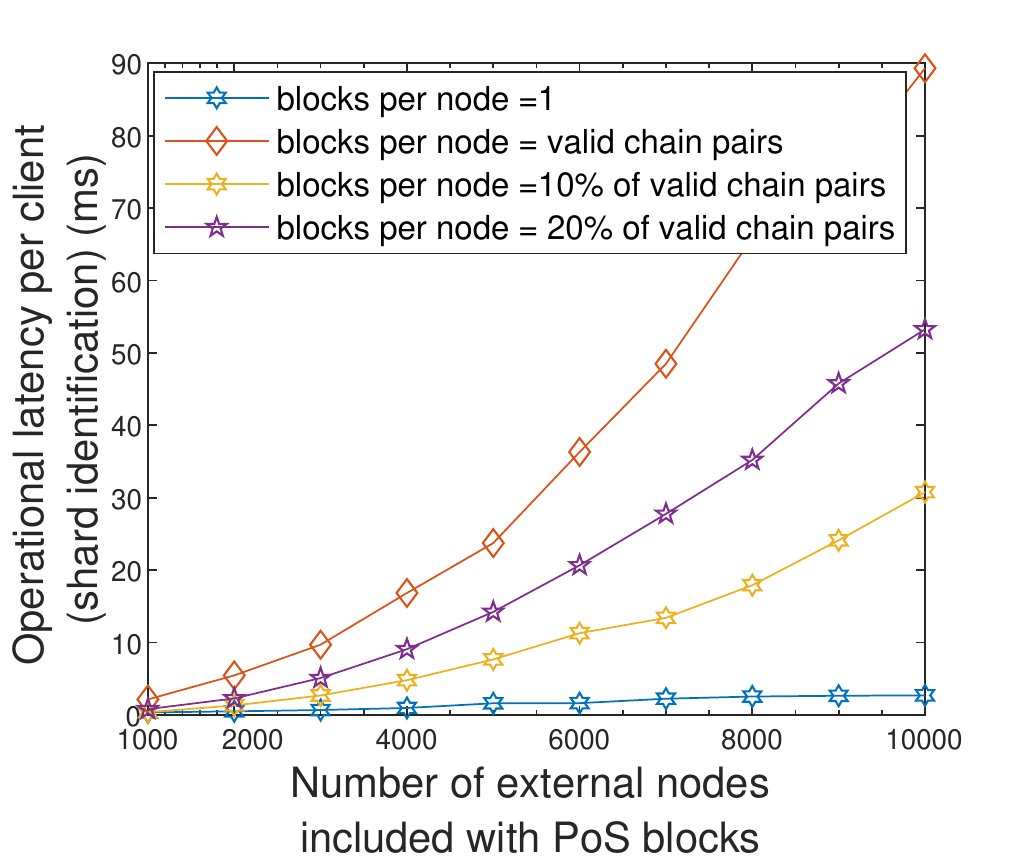}\\
        (a)\label{g1}
    \end{minipage}%
    \begin{minipage}[b]{.5\linewidth}
    \centering
        \includegraphics[width=1\linewidth]{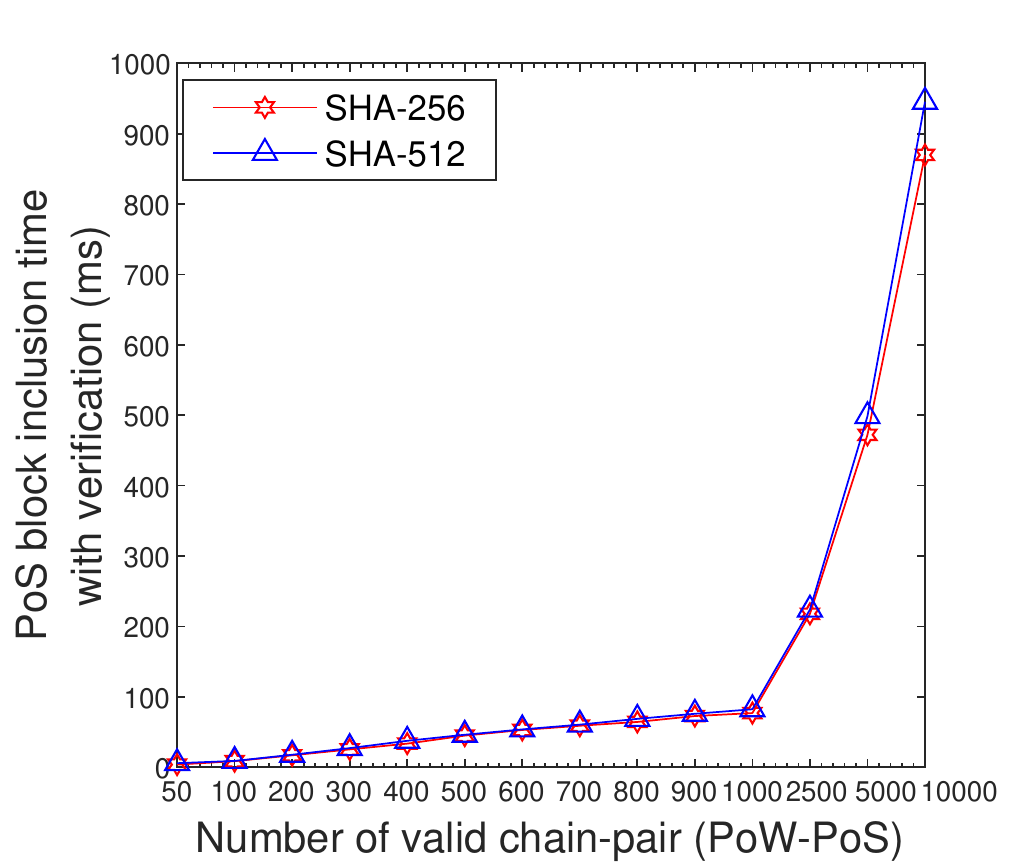}\\
        (b)\label{g2}
    \end{minipage}
    \caption{(a) Operational latency per client (shard identification) vs. the number of external nodes. (b) PoS block inclusion time vs. the number of valid chain-pairs.}\label{mg1}
\end{figure}
The time to shard and allocation decisions account for latency for which each node must wait before processing/generating new blocks in the blockchain. In the proposed Reinshard, the sharding is performed optimally by its architecture, which causes very few overheads that do not show a major impact on the performance. The result in Fig.~\ref{mg1}(a) suggests that an application-specific scenario (Scenario-1), where a node can be a part of only one shard, is efficient as sharding is done at the chain-pair, whereas for a scenario where a node can belong to multiple shards (Scenario-2 and Scenario-2') sharding may cause additional overheads, but these are well in the limits and do not affect the functioning of the entire blockchain. The increment of 20\% in the total blocks per node can increase the overheads by 46.6\%. However, the maximum range of latency is quite low at 89.32 ms. Furthermore, belonging to multiple shards enhances the types of applications for which \textit{Reinshard} can be used as well as help in maintaining high availability and accessibility in case of failures. These results suggest that the number of blocks per node impacts the performance as more overheads are observed in deciding the members of the shards. More the number of involved shards more is the participation, and higher are the overheads. However, these overheads grow at a rate of 4.1\% only when the number of shard occurrences is doubled for each node. With more valid chain-pairs, the number of verification and signing increases that increase the overall appending time of the PoS block, as shown in Fig.~\ref{mg1}(b). However, during evaluations, it was found that the number of blocks per external node does not impact the average generation time. Rather, the number of validators has more impact on it. Alongside, the signature size also causes some overheads, but results vary only by 6.2\% between SHA-256 and SHA-512 at $\pm$ 2 (ms) error for 100 runs.

\begin{figure}[!ht]
    \centering
   \begin{minipage}[b]{.5\linewidth}
        \centering
        \fbox{\includegraphics[width=0.9\linewidth]{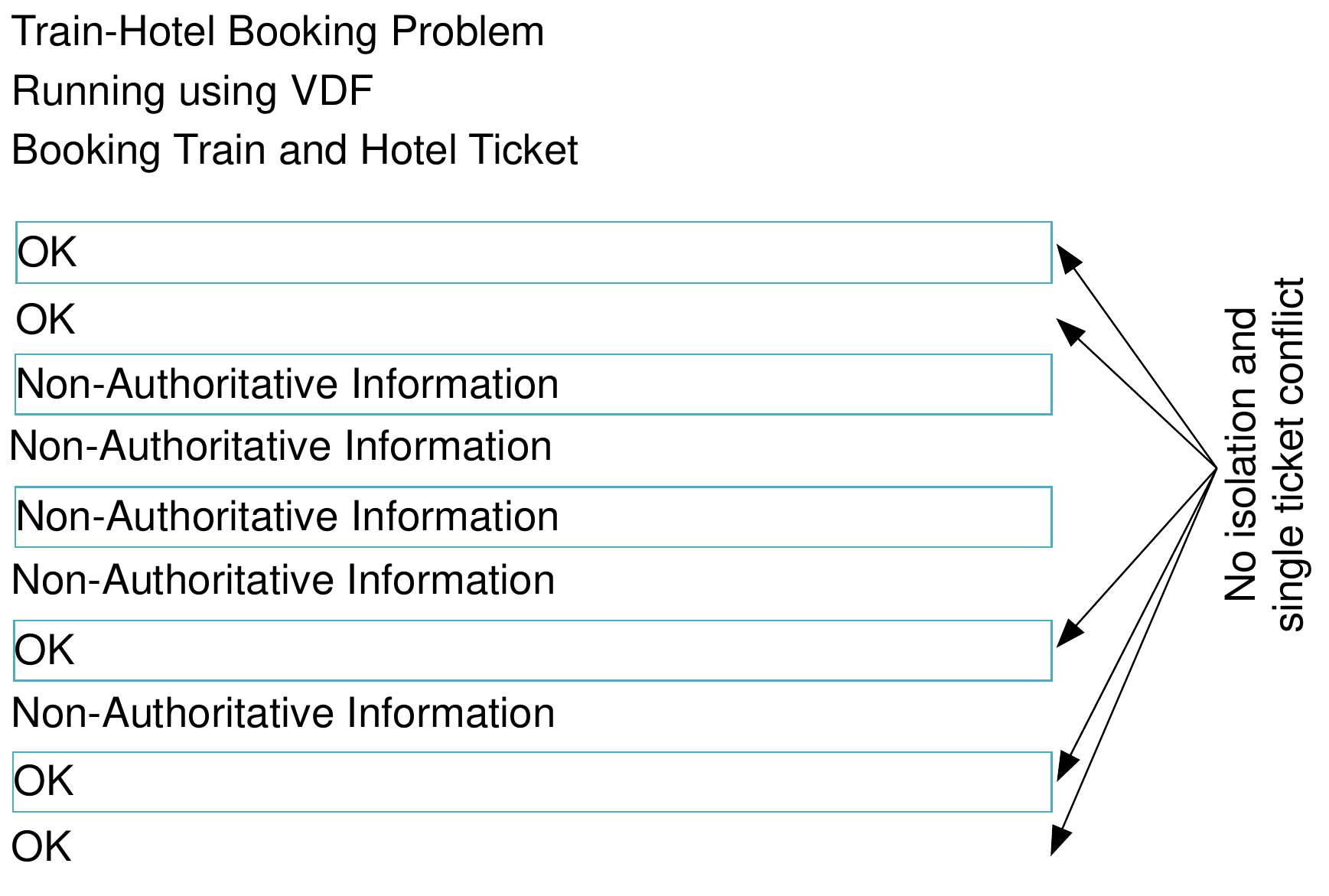}}\\
       (a)\label{g3}
    \end{minipage}%
    \begin{minipage}[b]{.5\linewidth}
    \centering
        \fbox{\includegraphics[width=0.92\linewidth]{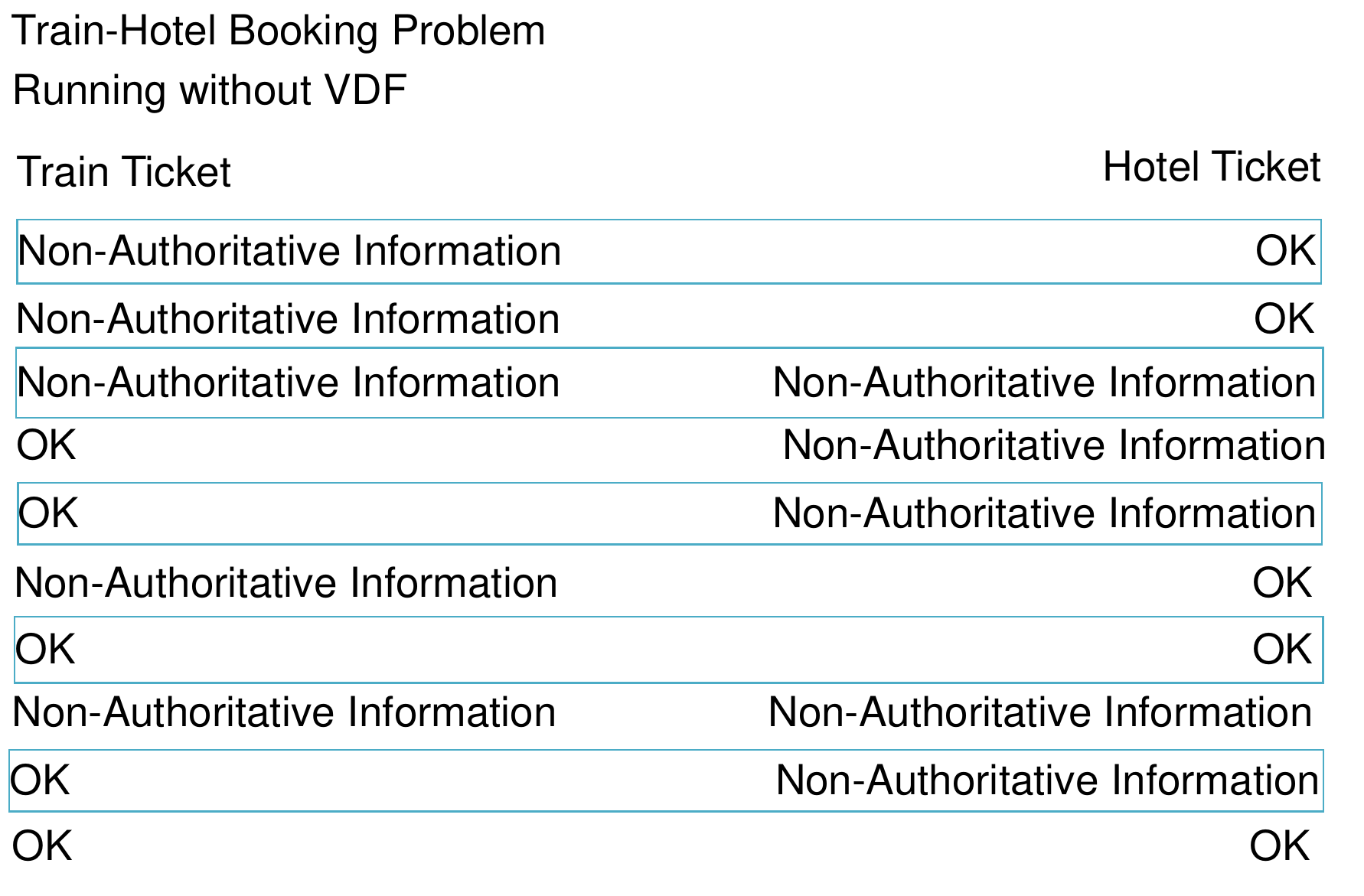}}\\
        (b)\label{g4}
    \end{minipage}
    \caption{(a) Full concurrent success with VDF. Nodes able to book both train and hotel tickets without isolation. (b). Operation failures without VDF. Nodes either having train or hotel tickets with randomized success.}\label{fig_chrome}
\end{figure}
\begin{figure}[!ht]
    \centering
   \begin{minipage}[b]{.5\linewidth}
        \centering
        \includegraphics[width=1\linewidth]{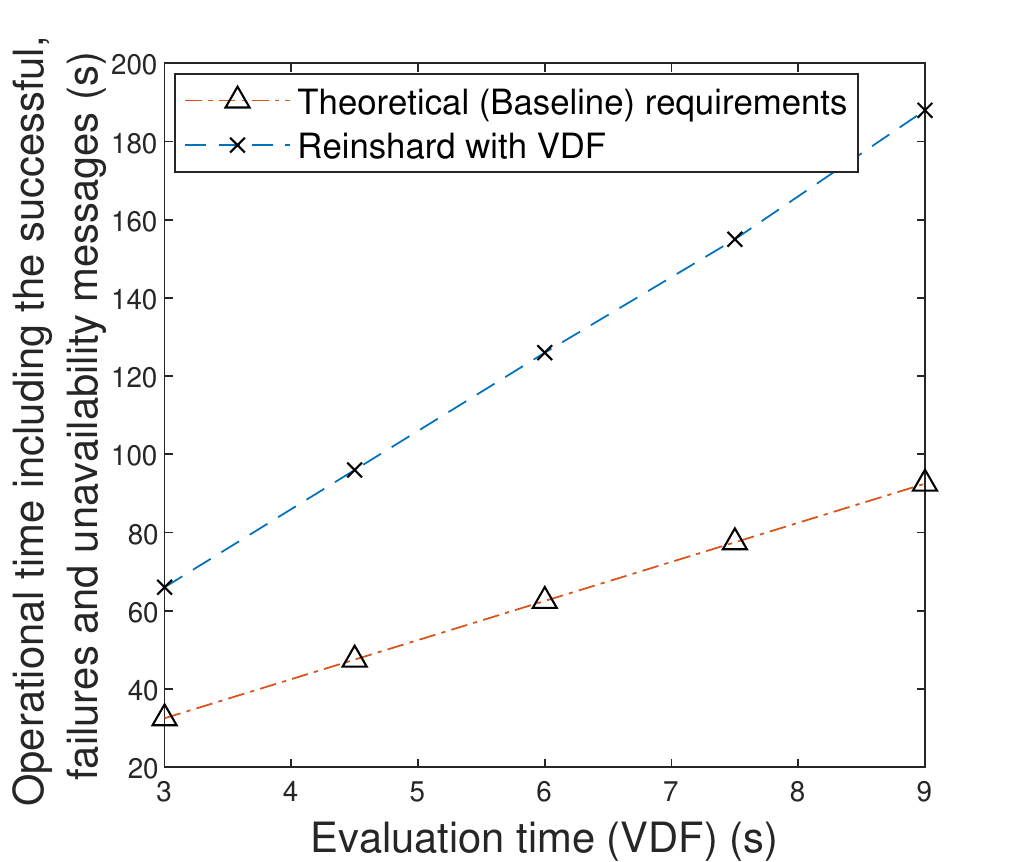}\\
       (a)\label{g5}
    \end{minipage}%
    \begin{minipage}[b]{.5\linewidth}
    \centering
        \includegraphics[width=1\linewidth]{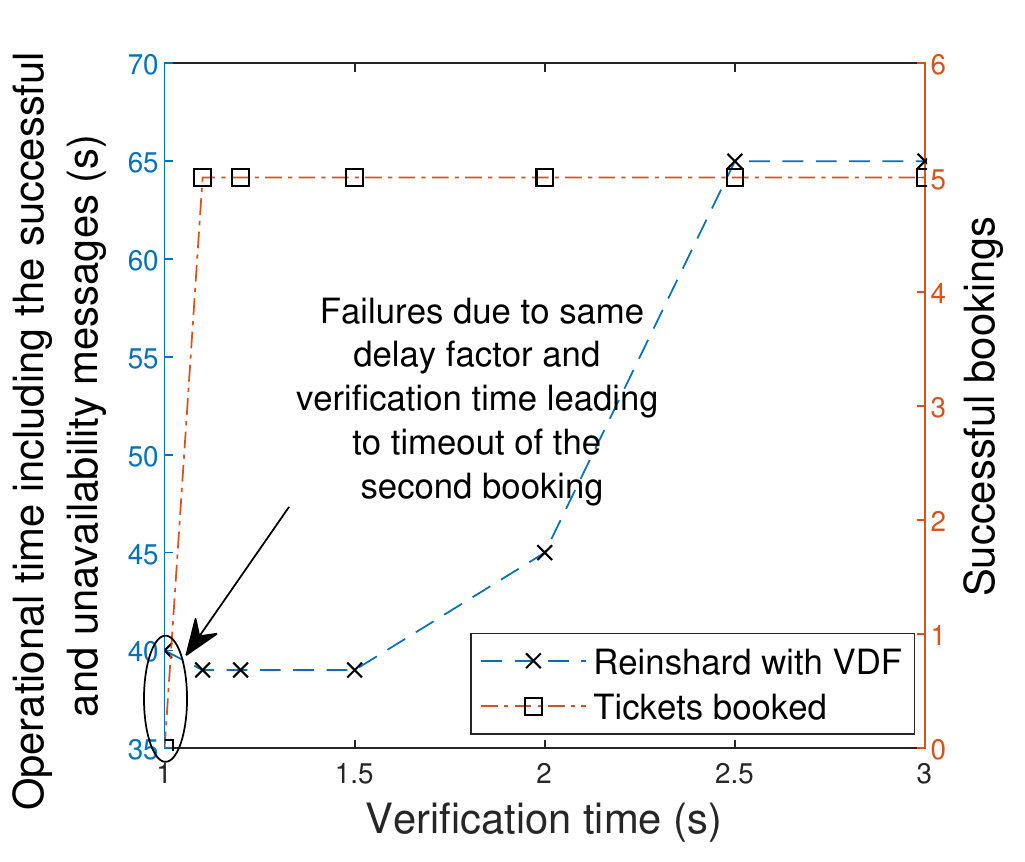}\\
        (b)\label{g6}
    \end{minipage}
    \caption{(a) Theoretical and observed values for Reinshard at 10 participants competing for tickets with a delay factor of 2.5s and verification time of 1s. (b) Operational overview with verification time of 1s and fixed evaluation time of 3s.}\label{mg2}
\end{figure}

\noindent \textbf{--Concurrency Resolution:} In \textit{Reinshard}, concurrency is resolved by using VDF as a delay factor that helps to maintain a lock (hold) on nodes. Such observations are presented by using the \textit{Train}-and-\textit{Hotel} booking problem~\cite{Tribble} with and without the use of VDF. To understand this, nodes acting as train and hotel booking servers are attached to the global blockchain by generating their respective PoS blocks and 50 concurrent requests are generated for each ticket. Both the servers belong to different shards. A simulated VDF is considered with a fixed value for the delay factor and is determined by the involved chain-pairs. At first, the number of available train and hotel tickets are deliberately kept at 5 each. It means that for efficiently resolving deadlocks in concurrent operations, the nodes with train tickets must be able to book the hotel as well. However, during run-time, it was observed that without the utilization of VDF, the nodes were able to generate random requests to either of the two ticketing servers resulting in some of the participants having one ticket only, as shown in Fig.~\ref{fig_chrome}(a). In contrast, the VDF models hold on the servers allowing nodes to establish connections with the train and hotel ticketing servers with the success of booking both the tickets, as shown in Fig.~\ref{fig_chrome}(b).

\begin{table*}[!ht]
\centering
\fontsize{8}{10}\selectfont
\caption{A comparison between the proposed solution and the related works (Computing Power (P), Stake (S)) (*the platform reserves 51\% of the balance to prevent attacks).}\label{table_comp_1}
\begin{tabular}{cccccccccc}
\hline
Blockchain & \begin{tabular}[c]{@{}c@{}}Block \\ proposal\end{tabular} & Type &  \begin{tabular}[c]{@{}c@{}}Attack \\ Resilience\end{tabular} & \begin{tabular}[c]{@{}c@{}}Cross-platform \\ integration\end{tabular} & \begin{tabular}[c]{@{}c@{}}Sharding \\   by architecture\end{tabular} & \begin{tabular}[c]{@{}c@{}}Concurrency \\ resolutions\end{tabular}\\
\hline
\hline
\begin{tabular}[c]{@{}c@{}}2-hop \\ blockchain~\cite{Duong2017}\end{tabular} & PoW-PoS & Hybrid & $>$50\% (P) & \xmark & \xmark & \xmark \\
TwinsCoin~\cite{Duong2018} & PoW-PoS & Hybrid &  $>$50\% (P) & \xmark & \xmark & \xmark \\
Hcash~\cite{Hcashteam2017} & PoW-PoS & Hybrid &  $>$50\% (S)* & \cmark & \xmark & \xmark\\
PeerCensus~\cite{decker2016bitcoin} & PoW-BFT & Hybrid & 33\% (P) & \xmark & \xmark & \xmark \\
Nxt~\cite{Nxt2014} & PoS & \begin{tabular}[c]{@{}c@{}}Chain-\\ based\end{tabular}  & $>$50\% (S) & \xmark & \xmark & \xmark \\
\begin{tabular}[c]{@{}c@{}}Proposed \\ (Reinshard)\end{tabular} & PoW-PoS & \begin{tabular}[c]{@{}c@{}}Hybrid \\ (Dual)\end{tabular} & $>$50\% (P), (S) & \cmark & \cmark & \cmark \\
\hline
\end{tabular}
\end{table*}

\begin{figure}[!ht]
\begin{center}
    \includegraphics[width=1\linewidth, trim = {5cm 6cm 7cm 5cm}, clip]{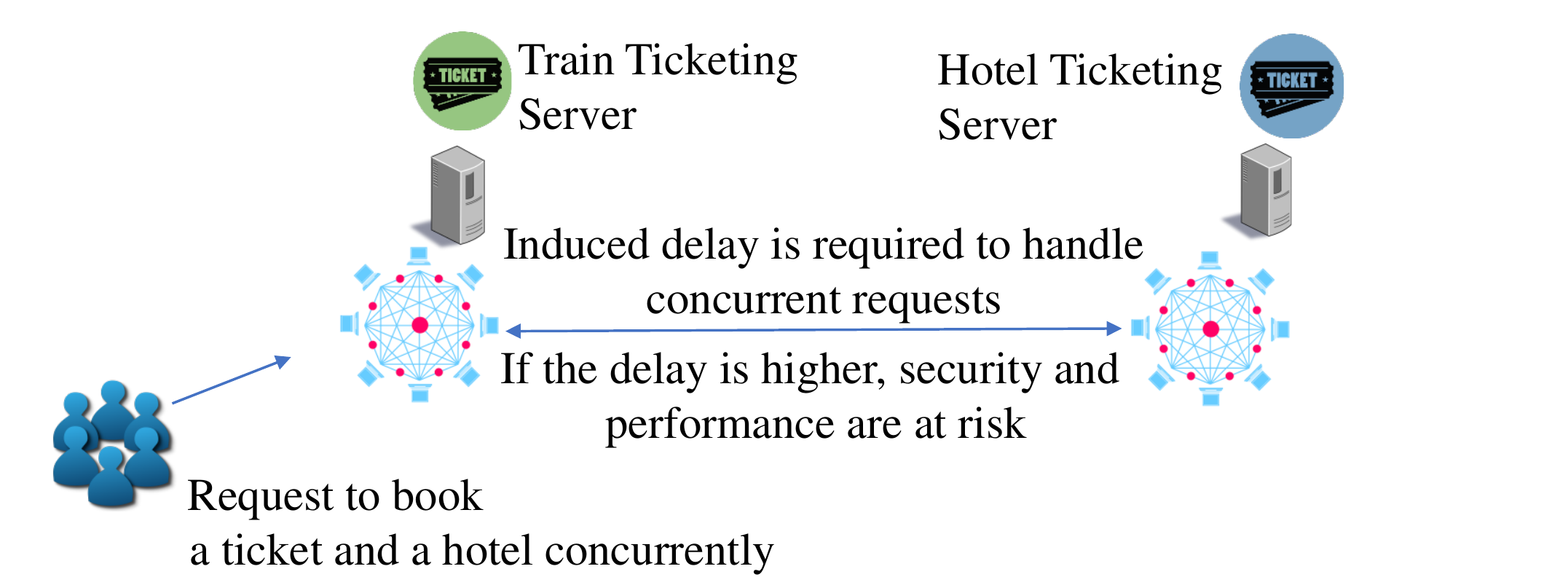}
    \end{center}
\caption{An illustration of concurrency problem in the blockchain.}
\label{Darkpool}
\end{figure}
The proposed model was also evaluated in the presence of an adversary under the Train-Hotel booking problem, as shown in Fig.~\ref{Darkpool}. For this, the adversaries were modeled with and without VDF capabilities and their possessions were varied between 30-50\% of the computational power and stakes (controlling the chain-pairs). The simulations were carried under the same settings. Despite such favorable conditions for an adversary (which is practically difficult to attain because of sequential and unique features of the VDF function), it could only lead to DoS attack but cannot affect the concurrent operations as no users were stranded with a single ticket.
\begin{figure}[!ht]
  \centering
 \fbox{\includegraphics[width=0.7\linewidth]{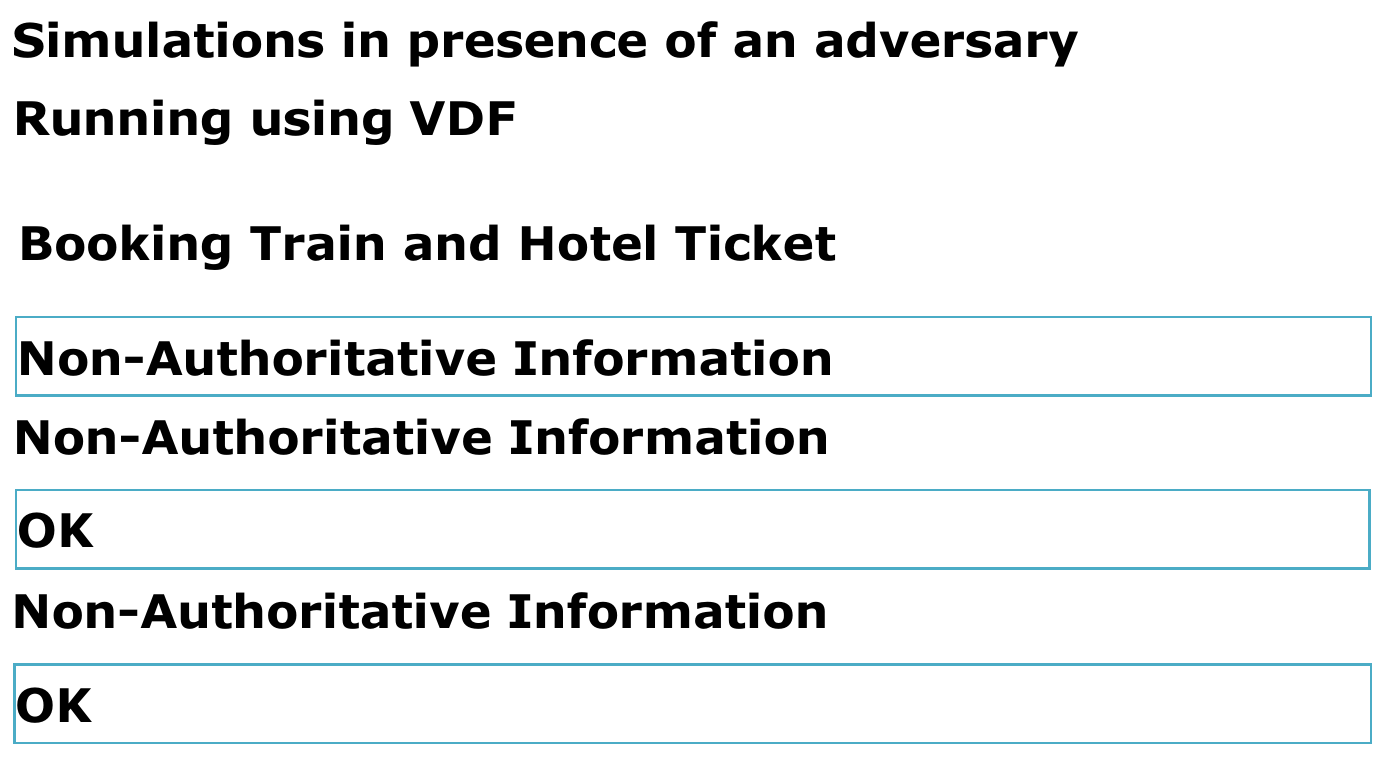}}
  \caption{Evaluations of concurrent operations when 30-50\% of the entire blockchain is controlled by an adversary under the Train-Hotel booking problem. No user was left with only one ticket. The only loss is observed in denial of tickets to some users, which is far more acceptable than the general scenario where under adversarial conditions, most of the users again ended with only of the two tickets only.}\label{appendix_fig}
\end{figure}
The key observations were the inability of an adversary to hinder concurrent operations even at the maximum capacity. This means under ideal assumptions; the adversary was able to prohibit the booking, but it could not cause serialization where one user could book only one ticket. The concurrency was not affected because of the presence of the validators (\textit{manages hold}), which are the part of PoW-PoS chain-pair and follows similar VDF puzzle-solving even in the shards. These observations can be visualized in Fig.~\ref{appendix_fig}. However, to fully control the entire blockchain, an adversary needs to be at the communicating chain-pairs as well as control more than 51\% of the PoS sub-chain, which is practically difficult as it involves solving an exponential time puzzle under the limits when other operations are being handled in a deterministic polynomial duration.

\begin{table*}[!ht]
\centering
\fontsize{7.5}{9}\selectfont
\caption{A comparison for attack prevention without centralization and third party evaluations (\textnormal{requires}- $^{\dag}$\textnormal{central entity/checkpoints}, $^{\ddag}$\textnormal{identity provider}).}\label{comp_table_2}
\begin{tabular}{cccccc}
\hline
Blockchain & \begin{tabular}[c]{@{}c@{}}Nothing-at-stake\end{tabular} & \begin{tabular}[c]{@{}c@{}}Long-range attack\end{tabular} & \begin{tabular}[c]{@{}c@{}}No Unwanted Centralization/Non-TEE executions\end{tabular} & \multicolumn{2}{c}{\begin{tabular}[c]{@{}c@{}}System level block generation\\ (Block Size: Time)\\ (KB: ms)\end{tabular}} \\
\hline
\hline
Ppcoin~\cite{king2012ppcoin,li2017securing} & \xmark & \cmark$^{\dag}$ & \xmark & \multicolumn{2}{c}{-} \\
Blackcoin~\cite{li2017securing,vasin2014blackcoin} & \xmark & \cmark$^{\dag}$ & \xmark & \multicolumn{2}{c}{-} \\
Snow White~\cite{daian2017snow,li2017securing} & \xmark & \cmark$^{\dag}$ & \xmark & \multicolumn{2}{c}{-} \\
\begin{tabular}[c]{@{}c@{}}TEE-based PoS~\cite{li2017securing}\end{tabular} & \cmark & \cmark$^{\ddag}$ & \xmark & \begin{tabular}[c]{@{}c@{}}1\\ 1000\end{tabular} & \begin{tabular}[c]{@{}c@{}}$\sim$1.70\\ $\sim$99-100\end{tabular} \\
NXT~\cite{Nxt2014,li2017securing} & \xmark & \cmark$^{\dag}$ & \xmark & \begin{tabular}[c]{@{}c@{}}1\\ 1000\end{tabular} & \begin{tabular}[c]{@{}c@{}}$\sim$0.42\\ $\sim$99-100\end{tabular} \\
Proposed& \cmark & \cmark & \cmark & \begin{tabular}[c]{@{}c@{}}1\\ 1000\end{tabular} & \begin{tabular}[c]{@{}c@{}}$\sim$0.39\\ $\sim$69.22\end{tabular}\\
\hline
\end{tabular}
\end{table*}
To further understand the handling capacity of the \textit{Reinshard} in the real-time, a delay factor is set around 2.5 s, which is the time to book one ticket, the evaluation time and verification time are ranged between 3 and 9 s, and 1 and 3 s, respectively. Where no concurrent control operations are used in \textit{Reinshard}, the allocation time is around 39 s for a short evaluation duration. However, as stated previously, such a scenario resulted in certain failures where a node is able to book only one ticket. Theoretically, a complete operation including unavailability messages must take 32.5 to 92.5 s with a varying evaluation time. This serves as our baseline as anything below this value means incorrect operations. The proposed \textit{Reinshard}, with VDF (sequences in booking a train first), performs accurately with only successful participants getting both the tickets and the entire process with unavailability messages was completed in the time between 66 and 188 s with varying evaluation time, as shown in Fig.~\ref{mg2}(a). The actual time to book the tickets, irrespective of the competing ones, is between 34 and 35 s. As per the architecture as well as the security analysis, no ticket should be booked when the verification time and the delay factor are the same. Such observations can be noticed in Fig.~\ref{mg2}(b). These traces show that at the same verification time and delay factor, the system goes into deadlocks and generates error messages without booking any ticket. However, as the significant difference between both increases, users can book all the available tickets in a competitive time.

\noindent\textbf{--Comparison with state-of-the-art}: Some related works, as discussed in Table~\ref{table_comp_1}, help to understand the key differences between the proposed (\textit{Reinshard}) and the existing consensus mechanisms at the architectural and property levels. It is clear that none of the existing solutions provide any direct resolution to shard the nodes especially for using blockchain beyond cryptocurrencies. Recently, it has been determined that the architectures that aim at scaling through sharding must be able to resolve concurrency issues amongst shards, which is neither provided nor resolved by any of the existing approaches such as 2-hop blockchain~\cite{Duong2017}, TwinsCoin~\cite{Duong2018}, Hcash~\cite{Hcashteam2017}, PeerCensus~\cite{decker2016bitcoin}, Nxt~\cite{Nxt2014} or Ppcoin~\cite{king2012ppcoin}. Apart from the Hcash and the proposed solution, no other blockchain even talks about cross-platform integration, which is the actual future of the blockchain systems. Even for 51\% attack resilience, \textit{Reinshard} covers both the computational power as well as stake owner-ships, whereas other hybrid mechanisms only rely on the computational power for attack resilience. For the block generation, the delay from the VDF only accounts when cross-shard and concurrent requests are involved in the blockchain, otherwise, there is no delay imposed on the nodes which prevent any overheads on the operations of the entire chain. Apart from these, \textit{Reinshard} can prevent nothing-at-stake and long-range attacks that too without the use of checkpoints or trusted hardware. The existing solutions, \cite{king2012ppcoin,Nxt2014,daian2017snow,li2017securing,vasin2014blackcoin}, partially resolve the long-range attacks as their key functionality is dependent on a central entity.

At the system-level, the block generation rate of \textit{Reinshard} was compared with the Nxt~\cite{Nxt2014} and Trusted Execution Environment (TEE)-based PoS~\cite{li2017securing}, as the latter improves the former for preventing the nothing-at-stake and long-range attacks, however, at the cost of unwanted centralization and dependence on an external identity provider. Results were compared at two block sizes of 1 KB and 1000 KB using SHA-256 where \textit{Reinshard} was simulated with 100 to 1000 validators. The evaluations (Table~\ref{comp_table_2}) suggest that Reinshard can generate 1 KB block in 0.39 ms and 1000 KB block in 69.22 ms with a maximum of 1000 validators, which is far better than the Nxt and TEE-based PoS as these reported a block generation of 0.42 and 1.72 ms for 1 KB, respectively, and around 100 ms (approx.) for 1000 KB block size. \textit{Reinshard}, with its dual-chain architecture, the unpredictability of block allocation, and VDF puzzle, is efficient and secure. If an attacker tries to stick for a duration and be a leader at a particular height, it has to become one of the PoW-PoS chain-pairs, whose difficulty is affected by RL-reward mechanisms and the non-probabilistic adjustments reduce any chances for becoming a leader if the node does not show consistent participation. Furthermore, counterfeiting the PoS sub-chain requires solving of VDF exponential puzzle in polynomial time that should be lesser than the verification time, which is practically beyond the limits of the current infrastructure. Thus, protecting \textit{Reinshard}, without requiring any third-party solutions, checkpoints or identify providers.
\section{Conclusion}\label{conclusions}
A dual blockchain, \textit{Reinshard}, via hybrid consensus was proposed, which is inspired by the existing 2-hop blockchain for optimal sharding with a high capacity of resolving concurrency issues in cross-shard communications. The proposed blockchain can provide a low complex and non-probabilistic difficulty adjustment in a non-flat model, which is far more suitable for practical applications than the probabilistic adjustments. The use of VDF helped to resolve concurrency issues as well as control the growth of the PoS sub-chain under the PoW-PoS chain-pair. The security proofs helped to validate the resilience of the proposed blockchain against known attacks as well as establish general blockchain properties. The experimental study helped to understand the practical aspect where the Reinshard is used for handling concurrent operations. In the future, the key target would be to scale the proposed Reinshard as an independent platform, which can provide a unique feature of combining different PoS-blockchains as sub-chains of the PoW-PoS chain pairs. Furthermore, a rigorous mechanism would be developed to provide a secure and efficient way of generating VDF for PoS-leader selection.
\bibliographystyle{ieeetr}
\bibliography{ref_reinshard}

\ifCLASSOPTIONcaptionsoff
  \newpage
\fi

\newpage
\begin{IEEEbiography}[{\includegraphics[width=1.07in,height=1.20in,clip]{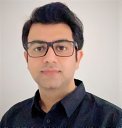}}]{Vishal Sharma}
received the PhD and BTech degrees in computer science and engineering from Thapar University (2016) and Punjab Technical University (2012), respectively. He is working as a Lecturer ($\sim$Assistant Professor) in the School of Electronics, Electrical Engineering and Computer Science (EEECS) at the Queen's University Belfast (QUB), Northern Ireland, United Kingdom. Before coming to QUB, he was a Research Fellow in the Information Systems Technology and Design (ISTD) Pillar at the Singapore University of Technology and Design (SUTD), Singapore where he worked on the future-proof blockchain systems funded by SUTD-MoE. From Nov'16 to Mar'19, he worked in the Information Security Engineering Department at Soonchunhyang University, South Korea in multiple positions (Nov'16 to Dec'17: Postdoctoral Researcher; Jan'18 to Mar'19: Research Assistant Professor). He also held a joint postdoctoral position with Soongsil University, South Korea. He was affiliated with the Industry-Academia Cooperation Foundation and the Mobile Internet Security lab at Soonchunhyang University. Before this, he worked as a lecturer in the Computer Science and Engineering Department at Thapar University, India. He is the recipient of three best paper awards from the International Conference on Communication, Management and Information Technology (ICCMIT), Warsaw, Poland, in April 2017; from CISC-S'17 South Korea in June 2017; and from IoTaas Taiwan in September 2017. He is a member of IEEE, a professional member of ACM and a past Chair for ACM Student Chapter-TIET Patiala. He has authored/co-authored more than 100 journal/conference articles and book chapters and co-edited two books with Springer. He serves as the ATE for the IEEE Communications Magazine and an AE for the IET-CAAI TRIT, Wireless Communications and Mobile Computing, and IET Networks. His areas of research and interests are 5G networks, Blockchain systems, aerial (UAV) communications, CPS-IoT behaviour-modelling, and mobile Internet systems.
\end{IEEEbiography}
\vskip 0pt plus -1fil
\begin{IEEEbiography}[{\includegraphics[width=1.07in,height=1.20in,clip]{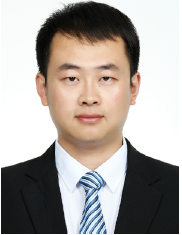}}]{Zengpeng Li}
is an associate research professor at Shandong University (SDU), China. Before joining SDU, he worked as a postdoctoral research fellow at the Singapore University of Technology and Design (SUTD) and a postdoctoral research associate at Lancaster University, UK. He received PhD from Harbin Engineering University. During his PhD program, he was a researcher assistant (a.k.a., long-term visiting PhD student) at the Virginia Commonwealth University and University of Auckland, respectively. His primary research interests are in cryptography and secure distributed computing, and specific topics include: compute on encrypted data, verifiable computation in outsourced environments, password-based cryptography, random beacon, and blockchain consensus.
\end{IEEEbiography}
\vskip 0pt plus -1fil
\begin{IEEEbiography}[{\includegraphics[width=1.07in,height=1.20in,clip]{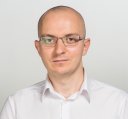}}]{Pawel Szalachowski}
was an Assistant Professor at Singapore University of Technology and Design (SUTD), Singapore. Before joining SUTD, he was a senior researcher at ETH Zurich. He received his PhD degree in Computer Science from Warsaw University of Technology in 2012. His current research interests include systems security and blockchains.
\end{IEEEbiography}
\vskip 0pt plus -1fil
\begin{IEEEbiography}[{\includegraphics[width=1.07in,height=1.20in,clip]{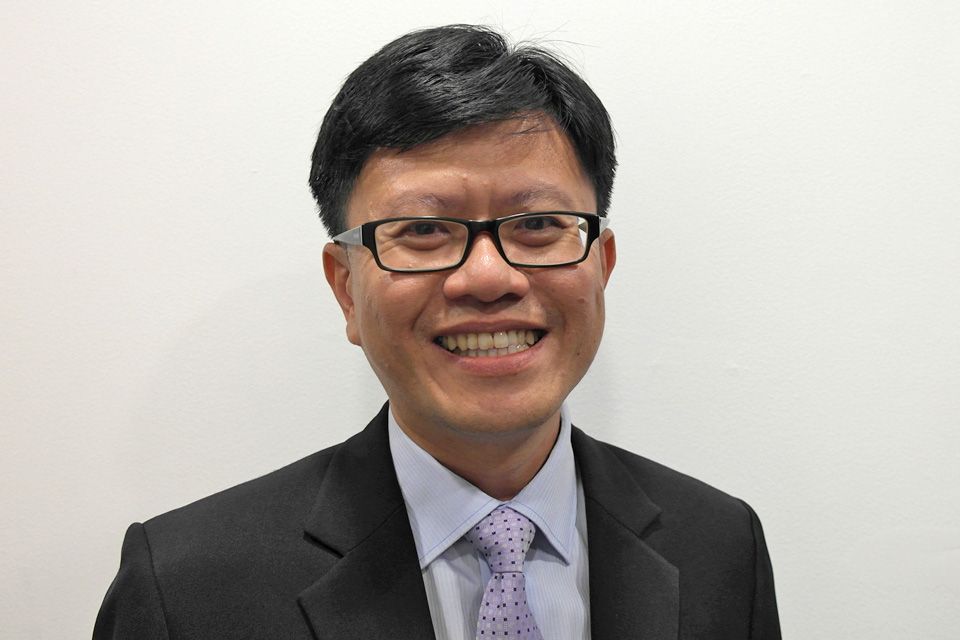}}]{Teik Guan Tan}
is working towards his PhD at Singapore University of Technology and Design (SUTD) under the supervision of Prof. Jianying Zhou. He completed his BSc and MSc from the National University of Singapore between 1992 and 1996. In his past 20 years of experience in the industry, he has designed and implemented banking and government security systems supporting millions of users and protecting billions of dollars worth of transactions. His research interests are in the intersection of authentication, applied cryptography and Quantum algorithms.
\end{IEEEbiography}
\vskip 0pt plus -1fil
\begin{IEEEbiography}[{\includegraphics[width=1.07in,height=1.20in,clip]{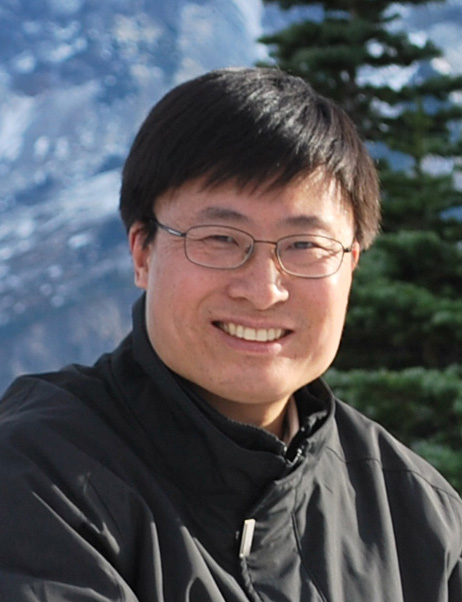}}]{Jianying Zhou}
is a professor and co-centre director for iTrust at Singapore University of Technology and Design (SUTD). Before joining SUTD, he was a principal scientist and the head of the Infocomm Security Department at Institute for Infocomm Research, A*STAR. He also worked at the headquarters of Oracle as a security consultant. Prof. Zhou received PhD in Information Security from Royal Holloway, University of London. His research interests are in applied cryptography and network security, cyber-physical system security, mobile and wireless security. He has published 250+ refereed papers at international conferences and journals with 10,000+ citations and received ESORICS'15 best paper award. He has 2 technologies being standardized in ISO/IEC 29192-4 and ISO/IEC 20009-4, respectively. He also has 10+ technologies being patented. Prof. Zhou is a co-founder \& steering committee co-chair of ACNS, which is ranked in the top 20 cybersecurity conferences. He is also steering committee chair of ACM AsiaCCS, and steering committee member of Asiacrypt. He has served 200+ times in international cybersecurity conference committees (ACM CCS \& AsiaCCS, IEEE CSF, ESORICS, RAID, ACNS, Asiacrypt, FC, PKC etc.) as general chair, program chair, and PC member. He has also been on the editorial board of top cybersecurity journals including IEEE Security \& Privacy, IEEE TDSC, IEEE TIFS, Computers \& Security. He received the ESORICS Outstanding Contribution Award in 2020, in recognition of contributions to the community.
\end{IEEEbiography}




\end{document}